\documentclass[a4paper,12pt]{amsart}
\usepackage{amssymb}
\usepackage{amsmath}
\usepackage{amsmath}
\usepackage{amsfonts}
\usepackage{amssymb}
\usepackage[utf8]{inputenc}
\usepackage[english]{babel}
\usepackage{enumitem}
\newlist{assumptions}{enumerate}{1}
\newlist{altassumptions}{enumerate}{1}
\setlist[assumptions]{label=(H\Alph*),ref=H\Alph*}
\setlist[altassumptions]{label=(H\Alph*$'$),ref=H\Alph*$'$}

\usepackage{hyperref}

\usepackage{thmtools}

\declaretheorem[name=Theorem,numberwithin=section]{thm}

\usepackage{url}

\DeclareMathOperator*{\supp}{supp}

\DeclareMathOperator*{\spanop}{span}

\newcommand{\Hi}{\mathcal{H}}

\newcommand{\Pro}{\mathbb{P}}

\newcommand{\R}{\mathbb{R}}
\newcommand{\C}{\mathbb{C}}

\newcommand{\N}{\mathbb{N}}

\newcommand{\Z}{\mathbb{Z}}

\renewcommand{\leq}{\leqslant}
\renewcommand{\geq}{\geqslant}
\renewcommand{\Re}{\operatorname{Re}}

\newcommand{\bottom}{E_\epsilon}
\newcommand{\finhamper}[1][\theta]{H_{q,\epsilon}^\square(#1)}

\newcommand{\findem}{\begin{flushright} $\blacksquare$ \end{flushright}}
\newcommand{\debdem}{\subsubsection*{Proof}}

\newtheorem{cor}[thm]{Corollary}
\newtheorem{lem}[thm]{Lemma}
\theoremstyle{remark}
\newtheorem{rem}[thm]{Remark}

\usepackage{appendix}
\usepackage{color}
\usepackage{graphicx}
\usepackage{srcltx}

\begin{document}


\title[Expansion of the almost sure spectrum]{Expansion of the almost sure spectrum in the weak disorder regime}
\author[D. Borisov]{Denis Borisov}
\address[D. Borisov]{Institute of Mathematics Of Ufa Scientific Center of RAS, Chernyshevskogo, 112, Ufa, 450008, Russia 
\newline \mbox{\quad} \& \newline
Department of Physics and Mathematics, Bashkir State Pedagogical University, October rev. st.~3a, Ufa, 450000, Russia}
\urladdr{http://borisovdi.narod.ru/}
\email{borisovdi@yandex.ru}

\author[F. Hoecker-Escuti]{Francisco Hoecker-Escuti}
\address[F. Hoecker-Escuti]{Technische Universit\"at Chemnitz, Fakult\"at f\"ur Mathematik, Reichenhainer Str. 41, Chemnitz, D-09126, Germany}
\curraddr[F. Hoecker-Escuti]{Technische Universit\"at Hamburg-Harburg, Am Schwarzenberg-Campus 3, D-21073 Hamburg , Germany}
\email{francisco.hoecker-escuti@tuhh.de}

\author[I. Veseli\'c]{Ivan Veseli\'c} 
\address[I. Veseli\'c]{Technische Universit\"at Chemnitz, Fakultät für Mathematik, Reichenhainer Str. 41, Chemnitz, D-09126, Germany}
\urladdr{http://www.tu-chemnitz.de/mathematik/stochastik/}

\begin{abstract}
  The spectrum of random ergodic Schr\"odinger-type operators is almost surely a deterministic subset of the real line.
 The random operator can be considered as a perturbation of a periodic one.
  As soon as the disorder is switched on via a global coupling constant, the spectrum expands.
  We estimate how much the spectrum expands at its bottom for operators on $\ell^2(\Z^d)$.
\end{abstract}
\maketitle

\tableofcontents

\section{Introduction}
Due to the self-averaging property of ergodic Schr\"odinger operators
the resulting spectrum is almost surely a fixed subset of the real
line.  If a random operator is a perturbation of a periodic
operator, it is of interest to know how the spectrum expands once we
switch on the disorder via a global coupling constant.  Apart from the
genuine interest to identify the location of the spectrum, this is
also of central importance when identifying energy regions 
corresponding to localized wavepackets.

Otherwise it may happen that one proves a Wegner
estimate, a Lifschitz tail bound or a similar statement related to
localization, and then later discovers that the considered energy
regime belongs to the resolvent set.

In this paper we consider an $\epsilon$-small random perturbation of a discrete translation-invariant operator and we study how the bottom of its spectrum behaves. 
By symmetry, similar estimates apply to the location of the maximum of the spectrum, in a weak disorder regime.
 To fix the ideas, let us introduce a prototypical example.
 Let $\Hi=\ell^2(\Z^d)$ and $\Delta_{\Z^d}:\Hi \to \Hi$ the (negative definite) discrete
Laplacian on $\Z^d$, i.e.
$$ \left( \Delta_{\Z^d} u \right) (n) := \sum_{|n-m|_\infty =1} \left( u(m) -u(n) \right).$$ 
We define the operator $H_0:\Hi \to \Hi$ by
$$H_0 := - \Delta_{\Z^d} + W,$$ 
where $W$ is the multiplication operator by a real-valued function, which we
also denote by $W$ and which we assume periodic with respect to the subgroup
$\gamma := N\Z^d$. 

Let  $\square := [0,N-1]^d \subset \Z^d$ and  $V^\square \in \ell^\infty(\Z^d)$ be a non-trivial, compactly
supported single-site potential satisfying
$$ \supp(V^\square) \subset \square.$$ Let $(\omega_k)_{k \in \gamma}$
be a sequence of non-trivial, bounded, independent, identically distributed
random variables. For the sake of the introduction,
assume that $\{-1,1\} \in \supp \omega_0 \subset [-1,1]$. From now on we denote by
$V_\omega: \ell^2(\Z^d)\to \ell^2(\Z^d)$ the diagonal operator
defined, for $f \in \ell^2(\Z^d)$, as
\begin{equation}
\label{eq:randomV}
(V_\omega f)(x) = \sum_{k \in N\Z^d} \omega_k V^\square(x-k) f(x).
 \end{equation}
To motivate our results, let us consider the following discrete
alloy-type random Schrödinger operator defined by
\begin{align}
\label{def:alloytypeoperator}
H_{\omega,\epsilon} := H_0 + \epsilon V_\omega.
\end{align}
Under the stated assumptions, this operator is ergodic, and thus there exists a set
$\Sigma_\epsilon \subset \R$ such that
$${\sigma(H_{\omega,\epsilon})} = \Sigma_\epsilon$$ with
probability 1 (see e.g. \cite{pastur297spectra}). From now on we refer to $\Sigma_\epsilon$ as the {\it
  almost-sure spectrum} of $H_{\omega,\epsilon}$. The best known
example of this kind of operators is the celebrated {\em Anderson
  model}, where $H_0$ is the discrete Laplacian on $\Z^d$ (i.e. $W
\equiv 0$), $V^\square = \delta_0$ and $N=1$. In this case, it is not hard to
see (\cite{pastur297spectra}) that the bottom of the spectrum of the perturbed operator
$\bottom := \inf \left( \Sigma_\epsilon \right)$ moves away from the
bottom of the spectrum of the free operator $E_0 := \inf \left( \Sigma_0
\right)$ as
\
\begin{equation*}
  \bottom = E_0 - \epsilon .
\end{equation*}  
If one
considers instead, for example, the dipole model, i.e. $V^\square = \delta_0 -
\delta_{e_1}$ with $e_1=(1,0,\ldots,0)$, it is proven in \cite{caoelgart2012} that
\
\begin{equation*}
\bottom \leq E_0 - C\epsilon^2 .  
\end{equation*}
 In this note
we study this question for a very general, wide class of operators (see assumptions in Section \ref{sec:assumption}). More precisely, we prove some upper
bounds of the quantity $\bottom - E_0$, which in
turns gives us information on the location of the spectrum of the
perturbed operator. We also discuss some partial results on the lower bound.

In order to state the result in this setting, we need to
consider the operator $H_0$ with $N\Z^d$-periodic boundary
conditions. Because of the translation invariance, the subspace of
$N\Z^d$-periodic functions in $\ell^\infty(\Z^d)$ is invariant under
the action of $H_0$. This subspace is $N^d$-dimensional, so that the
action of the operator corresponds to a matrix we denote by
\begin{equation}
\label{def:alloytypeoperatorperiodicbc}
H^\square_0 : \ell^2(\square) \to \ell^2(\square).
\end{equation}
We now state the result.

\begin{thm}
\label{thm:anderson}
Let $H_{\omega,\epsilon}$ be the alloy-type random Schrödinger
operator defined by \eqref{def:alloytypeoperator} and $\bottom$ the
bottom of its corresponding almost-sure spectrum. To the $N\Z^d$-periodic operator
$H_0$ we associate a Hermitian matrix $H_0^\square \in \C^{N^d \times N^d}$, defined as in
\eqref{def:alloytypeoperatorperiodicbc},  and we let $\psi_1 \in \ell^2(\square)$ be the
(unique normalized) positive ground state of $H_0^\square$. Define
$$A_1 := - \left\langle \psi_1, V^\square \psi_1 \right\rangle_{\ell^2(\square)}.$$ There exists $A_2 \leq 0$ such that for $\epsilon>0$ small enough
$$\bottom \leq E_0 + \epsilon A_1 +  \epsilon^2 A_2.$$
Furthermore, if $A_1 = 0$ then $|A_2|$ is non-zero and larger than the spectral gap of $H_0^\square$, i.e. the difference
between its two smallest eigenvalues. 
\end{thm}
We provide an explicit formula for the constant $A_2$ only in the next section as it requires the introduction of additional notation. 
We have an analogous estimate for (fibers of) periodic operators, see theorems \ref{thm:main} and \ref{thm:mainaltversion}. 
In fact, the estimate for periodic operators is one step in the proof of theorem
\ref{thm:anderson}. In the context of periodic operators we have a
related, complementary lower bound, see lemmas \ref{lem:convinfspec}
and \ref{lem:pos-convinfspec}.

We would like to make some remarks on the relevance of this
result. First, the location of the bottom of the spectrum
with respect to the coupling constant has been the subject of several 
papers: with periodic potentials in dimension one
\cite{titchmarsh1946} and in arbitrary dimension
\cite{kirschsimon1987}, \cite{colindeverdiere1991}, as well as with
random positive potentials \cite{kirschetal1998localization} and under
some generic assumptions on $W$ \cite{klopp2002b}.
Recently, for non-negative perturbations, 
but without requiring the potential to be periodic or ergodic,
a result on the lifting of the bottom of the spectrum was given in theorem 1.1 of 
\cite{ElgartK-14}, see also \cite{BoutetdeMonvelLS-11} for a general lemma
on the lifting of the spectral infimum. These results have a slightly different scope than our result, 
where we ask how much the spectrum expands into the negative half axis. 
The latter question was studied before for Schr\"odinger operators in the multidimensional continuum space, e.~g.~in
\cite{GesztesyGS-92}, in particular for periodic potentials satisfying certain differentiability conditions.

Understanding the
spectrum provides valuable information on the solutions of partial
differential equations. In particular, if one considers the
Schrödinger equation for the Hamiltonian $H_{\omega,\epsilon}$, the
spectral type of the Anderson model characterizes the transport
properties of the underlying disordered medium. For this model, the
spectrum is expected to exhibit a transition from localized states at
the bottom of the spectrum (pure point spectrum with exponentially
localized eigenfunctions) to extended states (absolutely continuous
spectrum) in the bulk of the spectrum. This {\em Anderson transition}
is still a conjecture in the setting of this article. The existence of
localized states at the bottom of the spectrum has been studied in
many papers.
We invite the reader to consult the monographs
\cite{carmona1990spectral}, \cite{pastur297spectra},
\cite{stollman01}, \cite{hislop2008lectures} and their extensive
bibliography. The perturbative regime $\epsilon \ll 1$ has attracted
much attention \cite{aizenman1994localization}, \cite{wang2001},
\cite{klopp2002a}, \cite{klopp2002b}, \cite{elgart2009lifshitz}, \cite{caoelgart2012}, \cite{hoeckerescuti2013}, \cite{hoeckerescuti2014}, \cite{elgartetal2011locnonmono}, \cite{borisovveselic2011}, \cite{borisovveselic2013}. In this
regime one can prove very precise estimates of the interval of
localization, namely that states with energies in
\
\begin{equation*}
  I_\eta(\epsilon) := (-\infty, -C_0 \epsilon^\eta] \cap
    \Sigma_\epsilon = (-\infty, -C_0 \epsilon^\eta] \cap [E_\epsilon,
        + \infty)
\end{equation*} are localized.   
In \cite{elgart2009lifshitz} it was proved that in dimension
$d=3$ one may take $\eta$ to be as large as $2$ and in \cite{hoeckerescuti2014} that for $d\geq2$ this holds with $\eta<2$. These results are
meaningful, as for the Anderson model $\bottom = - C_1 \epsilon$. If we
now consider different potentials, we may have a quadratic expansion
of the bottom of the spectrum $E_\epsilon$, and understanding where
the bottom of the spectrum lies appears to be crucial, so that the
interval of localization is non-trivial. Some of the issues addressed
in this note were already explored in \cite{klopp2002b} where 
it is assumed that the single-site potential has a non-zero mean and the
Floquet eigenvalues of the underlying periodic discrete Schr\"odinger operator $H_0= -\Delta +W$ are
assumed to be non-degenerated, as well as in \cite{caoelgart2012} for
the dipole potential. 
These are special cases of our models. The general operator we study corresponds roughly to tridiagonal block matrices of the form
\
\begin{equation}\label{ex:tridiagonal}
  \
  \left(\begin{matrix}
   \ddots & \ddots & \ddots & & \\
   {B^*} & A  & B & 0 &\ddots \\
   0 & {B^*} & A  & B  & 0 \\
    \ddots & 0 & {B^*}  & A & B \\
   & & \ddots & \ddots & \ddots  
  \end{matrix}\right)
  +
  \left(\begin{matrix}
   \ddots & \ddots & \ddots & & \\
   0 & \omega_{n-1} V^\square & 0 &\ddots &\\
   \ddots & 0 & \omega_{n} V^\square  & 0  &\ddots \\
    &  \ddots & 0  & \omega_{n+1} V^\square & 0 \\
   & & \ddots & \ddots & \ddots  
  \end{matrix}\right),
\end{equation}
where $A$ and $V^\square$ are Hermitian matrices and $\{\omega_n\}$ i.i.d. random variables.
We introduce in Section \ref{sec:assumption} the general framework in
which our results are obtained. 

To complete the description of the obtained results, let us briefly address the question of the optimality of the  lower bound (or at least its exponent). 
As far as the authors know, there is no general result in the literature in this direction
(but note the previously mentioned \cite{ElgartK-14}, \cite{GesztesyGS-92} and that for the Anderson model the bottom of the spectrum is known explicitly). 
One may naively expect, from perturbation theory, that the behavior should be linear or quadratic. 
The question turns out to be more subtle as the behavior may depend on the speed at which the Floquet eigenvalues associated to the bottom of the spectrum approach their minimum, 
as the following example shows.
\
\begin{thm}\label{thm:interesting}
  For $H_0 := (-\Delta_{\Z})^2$ defined on $\ell^2(\Z)$ and $V^\square$ the multiplication operator given by the following single-site potential:
  \
  \begin{align*}
    V^\square &:= -\frac{1}{2} \delta_{-1} + \delta_0 - \frac{1}{2} \delta_1.
  \end{align*}
Set as before $H_{\omega,\epsilon} := H_0 + \epsilon V_\omega$, cf.~\eqref{eq:randomV} and \eqref{def:alloytypeoperator}.
  Pick some $\xi>1/4$. Then for $\epsilon>0$ small enough we have
  \
  \begin{align}
    E_\epsilon:= \inf \sigma( H_0 + \epsilon V_\omega ) \leq -\frac{1}{6} \epsilon^{1+2\xi}.
  \end{align}
\end{thm}
For this example, which is of the form \eqref{ex:tridiagonal}, the coefficient $A_1$ corresponding to the linear term vanishes. The bound in theorem \ref{thm:interesting} is nevertheless better than quadratic thanks to the quartic behavior of the Floquet eigenvalues in a neighbourhood of their minimum. 
Unfortunately, apart from the trivial linear bound, we have no corresponding lower bound, although some results in this article provide a first step in this direction. 

This work can be extended in several directions. It would be very interesting to find the corresponding lower bounds, or at least conditions under which the infimum of the spectrum does not expand linearly. 
A related question concerns the expansion of the spectrum near a band edge, where one can also prove Anderson localisation. 
Indeed, discrete alloy type models exhibit a spectrum with band structure, cf.~\cite{ElgartKTV-12}.
Note that if one studies the expansion of the spectrum from a band edge instead of the bottom of the spectrum, the Floquet eigenvalues may vanish faster than quadratically when approaching the edge, even for the operator defined by \eqref{def:alloytypeoperator}. Rather than a pathological example, theorem \ref{thm:interesting} provides a model for this situation. 
Another question of interest is the study of overlapping single-site potentials. Under some non-degeneracy condition (see remark \ref{rem:overlap}) the results stated here can be extended to this situation, but a full understanding needs to consider periodic approximations of every order, something we also believe necessary 
to prove the  lower bounds complementing theorem \ref{thm:anderson}.

In a forthcoming project we consider the same questions for operators of Schr\"odinger type in the continuum setting, i.e. for operators 
acting on (dense subspaces) of $L^2(\R^d)$. Most of our findings are similar. In the continuum, it is more natural to define the operators via quadratic forms, and then formulate appropriate regularity conditions. 
Also, certain additional compactness arguments are necessary, due to the infinite dimensionality of the Hilbert space over the periodicity cell. 
On the other hand, in specific situations, better results are possible in the continuum setting, due to unique continuation principles for solutions of partial differential equations.

\section{General model}\label{sec:assumption}

 Let $d \geq 1$ be the space dimension, $D=\Z^d$ be the physical space and
  $\gamma = N \Z^d$ a sub-lattice of $D$. We denote by $\square$ its periodicity cell,
  i.e. $\square := [0,N-1]^d \cap \Z^d$. Note that $ D = \bigcup\limits_{k \in
    \gamma} \{ x \in D : x-k \in \square\}$. We also denote the
  reciprocal periodicity cell as $\square^*:=[0,\frac{2\pi}{N})^d$. From now on we assume the following hypotheses to hold.
\begin{assumptions}

\item\label{hypa} Let $H_0:\ell^2(D) \to \ell^2(D)$ be a bounded,
  non-negative Hermitian  operator defined by the matrix
$$H_0 := \left( H_0(k,k') \right)_{k,k' \in D},$$
satisfying the following properties:
\begin{itemize}
\item for all $k,k' \in D$, we have $H_0(k,k') = \overline{H_0(k',k)}$;
\item there exists $k_0 \neq 0$ such that $H_0(0,k_0) \neq 0$;
\item the associated operator is $\gamma$-invariant, i.e. for every $k \in \gamma$
$$\left\langle \tau_k u, H_0 \tau_k v \right\rangle = \left\langle u,
  H_0 v \right\rangle,$$ where $u,v \in \ell^2(D)$ and $\tau_k$ is the
  translation by $k \in \gamma$ operator; and
\item the associated operator is of {\it finite hopping range} with hopping range $R$, i.e. 
   if $|k-k'| \geq R$ then
  $$H_0(k,k')=0.$$ 
\item Through a global energy shift we may assume, with no
loss of generality, that $E_0 := \inf \sigma(H_0)=0$.
\end{itemize}
Note that if an operator is of finite hopping range with hopping range
$R$, for some $R>0$, then it also is of finite hopping range with hopping range $R'$ for any
$R' >R$. On the other hand, any $\gamma$-invariant operator is also
$n\gamma$-invariant, $n \in \N$. This means that we can always assume
that $R=N$, without loss of generality.

\item\label{hypb} Let $V^\square:\square \to \R$ be a non-trivial
  Hermitian matrix (we call it the {\em single-cell potential}, even
  when $V^\square$ is not diagonal). For any bounded sequence
  $(\omega_k)_{k\in \gamma}$ of real numbers, we define the block diagonal operator
  \begin{align*}
    V_\omega &:  \ell^2(D) \to
  \ell^2(D) \\
    V_\omega &:= \sum_{k \in \gamma} \omega_k \tau_{-k} V^\square \tau_k.
  \end{align*}
\end{assumptions}
For any real number $q \in \R$, we denote also by $q$ the constant
sequence indexed by $\gamma$, equal to $q$ on every site in $\Z^d$. We thus
have, for example, that
\
\begin{equation} \label{def:periodicVq}
  V_q := q \sum_{k \in \gamma} \tau_{-k} V^\square \tau_k
\end{equation} 
and $V_q$ is $\gamma$-invariant.  

From now on, the values of $\omega$ will be drawn
  from a sequence of bounded, non-trivial, independent and identically
  distributed random variables with distribution measure $\mu$. We will write $ S_\mu := \supp \mu $ and we assume that 
\begin{equation*}
  \{s_-, s_+\} \in  S_\mu \subset [s_-,s_+],
\end{equation*}
where $s_-$ and $s_+$ satisfy one of the following alternatives:
\begin{assumptions}[start=3]
\item\label{hypc}  The random variables change sign, i.e. $s_- < 0 < s_+$.
\end{assumptions}
\begin{altassumptions}[start=3]
\item\label{hypcprime}  The random variables are positive, i.e. $0\leq s_- < s_+$. 
\end{altassumptions}
The methods in this paper may also be adapted to negative random variables.
\
\begin{rem}
  It looks tempting, in order to achieve $s_-=0$,  to renormalize the random variables by adding and substracting some periodic potential, but in this case the underlying non-random operator depends on $\epsilon$. 
On the other hand, it is indeed allowed to rescale the random variables by absorbing the scaling factor in the single site potential $V^\square$.
\end{rem}
Let us now define our object of study. For each $\epsilon>0$, we let
$$H_{\omega,\epsilon} := H_0 + \epsilon V_\omega$$ which is a
self-adjoint, ergodic operator. We denote its almost-sure spectrum by
$\Sigma_\epsilon$ and by 
\begin{equation} \label{def:bottomspec}
  \bottom := \inf \Sigma_\epsilon
\end{equation}
the bottom of the spectrum. We also write $H_{q,\epsilon} := H_0 + \epsilon V_q$ the
corresponding operator with $V_\omega$ replaced by the periodic potential $V_q$ (defined as in \eqref{def:periodicVq}) 
and $E_{q,\epsilon}:=\inf \sigma(H_{q,\epsilon})$. In the following we
will study the bottom of the spectrum $\bottom$ of the random
operator for small $\epsilon$.

We define a finite dimensional matrix associated to the above
objects. Define the $\left( \left| \square \right| \times \left|
\square \right| \right)$-matrix $H^\square_0(\theta)$ by its
coefficients
\begin{align}\label{def:floquetmatrix}
\left( H_0^\square(\theta) \right)(k,k')  
&:= \sum_{m \in \gamma} e^{i\theta \cdot m} H_0(k,k'-m)  \\
& \phantom{:}= \sum_{\substack{|m|\leq N \\ m \in N\Z^d}} e^{i\theta \cdot m} H_0(k,k'-m), \nonumber
\end{align}
where $k,k' \in \square$.
Note that the second line is a consequence of the finite hopping
range and the sum in \eqref{def:floquetmatrix} is thus finite. Now define the matrix
$\finhamper$ by
$$\finhamper  := H_0^\square(\theta) + \epsilon q V^\square.$$
\begin{rem}
  The matrix $H^\square_{q,\epsilon}$ represents the action of $H_{q,\epsilon}$ on the fiber of $\theta$-quasi\-periodic functions in the Floquet--Bloch direct integral decomposition. More precisely, let (abusing notation)   $\varphi \in \ell^2(\square) \subset \ell^2(\Z^d)$. Then, regarding $H_{q,\epsilon}$ as an operator $\ell^\infty(\Z^d) \to \ell^\infty(\Z^d)$,
\
\begin{equation}\label{op:quasiperiodization}
  H_{\epsilon,q}^\square(\theta) \varphi = \chi_{\square}  H_{\epsilon,q} \bigl(\sum_{m \in \gamma} e^{i\theta \cdot m} \tau_m \varphi \bigr) \in \ell^2(\square), \quad \text{(abusing notation)}
\end{equation}
where $\chi_{\square}$ is the indicator function of $\square \subset \Z^d$.
\end{rem}

\section{Main results}

Recall that, by the continuity of the Floquet--Bloch eigenvalues (\cite{krueger2011floquet}, \cite{reedsimon01}),
there exists some $\theta$ such that
$$\inf \sigma(H_0^\square(\theta)) = E_0:=\inf \sigma(H_0) = 0.$$ 
We
denote by $\Theta \subset \square^*$ the compact set of $\theta$ for
which the last equality holds. From now on we fix some $\theta \in
\Theta$, so the quantities below will depend on $\theta$. Let
$\mathcal V_0 $ be the eigenspace of $H_0(\theta)$ associated to the
eigenvalue $E_0 = 0$ and $p$ its multiplicity. Choose an orthonormal
basis $\psi_j$, $j=1,\ldots,p$ spanning $\mathcal V_0$ and
diagonalizing the Hermitian matrix $A \in \C^{p \times p}$, given by
the coefficients
$$A_{ij} := \left\langle \psi_i, V^\square\psi_j \right\rangle  \quad
\text{for } 1\leq i,j\leq p.$$ 
We denote the eigenvalues $P_1\leq \ldots \leq P_p$ of the matrix $A$ in ascending order counting
multiplicities.
We would like to emphasize that the vectors  $\psi_j$, $j=1,\ldots,p$ are not eigenvectors of the operator $V^\square$ but diagonalize the self-adjoint operator $R_{\mathcal V_0}^* V^\square R_{\mathcal V_0}$ viewed as an endomorphism of $\mathcal V_0$, where $R_{\mathcal V_0}$ is the projection onto $\mathcal V_0$. In particular, we have the following
\begin{equation}\label{eq:orthogonalityrelations}
  \begin{aligned}
    \left\langle  V^\square \psi_i, \psi_j \right\rangle = P_i \delta_{ij}   \quad
\text{for } 1\leq i,j\leq p,
  \end{aligned}
\end{equation}
where $\delta_{ij}$ denotes Kronecker's delta.

Without loss of generality we assume that the orthonormal
basis $\psi_j$ of eigenvectors was enumerated in such a way that $P_1=A_{11} = \left\langle \psi_1,
V^\square\psi_1 \right\rangle $ and $P_p=A_{pp} = \left\langle \psi_p, V^\square\psi_p \right\rangle $.

Our result for sign-changing random variables reads as follows.
\begin{thm}
\label{thm:mainintro}
Assume \eqref{hypa}, \eqref{hypb} and \eqref{hypc}. Fix $\theta \in \Theta$ and define
\begin{equation} \label{def:A1}
  A_1:= \inf_{q \in S_\mu} \inf_{\substack{\psi \in \mathcal V_0
    \\ \|\psi\|_{\ell^2(\square)} = 1}} q \left\langle \psi,
  V^\square \psi \right\rangle = \min( s_+ P_1, s_- P_p) \leq 0,
\end{equation}
and
\begin{equation} \label{def:A2}
A_2:= - \max(s_-^2,s_+^2) \sup_{\substack{\psi \in \mathcal V_0
    \\ \|\psi\|_{\ell^2(\square)} = 1}} \sup_{\substack{\varphi \in \mathcal V_0^\perp
    \\ \|\varphi\|_{\ell^2(\square)} = 1}} \frac{\left| \left\langle \psi,
  V^\square \varphi \right\rangle \right|^2}{\left\langle H^\square_0(\theta)
  \varphi, \varphi \right\rangle} \leq 0.
\end{equation}
For any $\epsilon>0$ small enough the following holds: if $A_1 \neq 0$,
\begin{equation}
  \bottom \leq \epsilon A_1, \nonumber
\end{equation}
whereas if $A_1 = 0$, but $A_2\neq 0$,  then
\begin{equation}
  \bottom \leq \epsilon^2 A_2 + O(\epsilon^3). \nonumber
\end{equation}
Finally, if $A_1 = A_2 = 0$, then
$$\bottom \leq 0.$$
\end{thm}
Our result for positive random variables reads as follows.
\begin{thm}
\label{thm:mainaltversionintro}
Assume \eqref{hypa}, \eqref{hypb} and \eqref{hypcprime}. Fix $\theta \in \Theta$. 
Let us define the subspace $\mathcal V_{01} \subset \mathcal V_{0}$ as
\
\begin{equation*} 
  \mathcal V_{01} := \spanop {\{\psi_i  :i \in \N, P_i = P_1\}}  ,
\end{equation*}
i.e. the eigenspace of $A$ associated to its minimal eigenvalue $P_1$.
Define
\begin{equation} \label{def:Aprime1}
  A'_1:= \inf_{q \in S_\mu} \inf_{\substack{\psi \in \mathcal V_{01}
    \\ \|\psi\|_{\ell^2(\square)} = 1}} q \left\langle \psi,
  V^\square \psi \right\rangle = \min(s_+ P_1, s_- P_1) \in \R,
\end{equation}
and
\begin{equation} \label{def:Aprime2}
A'_2:= - s_+^2 \sup_{\substack{\psi \in \mathcal V_{01}
    \\ \|\psi\|_{\ell^2(\square)} = 1}} \sup_{\substack{\varphi \in \mathcal V_0^\perp
    \\ \|\varphi\|_{\ell^2(\square)} = 1}} \frac{\left| \left\langle \psi,
  V^\square \varphi \right\rangle \right|^2}{\left\langle H^\square_0(\theta)
  \varphi, \varphi \right\rangle} \leq 0.
\end{equation}
For any $\epsilon>0$ small enough the following holds: if $P_1 \neq 0$,
\begin{equation}
  \bottom \leq \epsilon A'_1, \nonumber
\end{equation}
whereas if $P_1 =A'_1= 0$, but $A'_2 = 0$, then
\begin{equation}
  \bottom \leq \epsilon^2 A'_2 + O(\epsilon^3). \nonumber
\end{equation}
Finally, if $P_1 = A'_1=A'_2 = 0$, then
$$\bottom \leq 0.$$
\end{thm}
Note that $A_1$ in theorem \ref{thm:mainintro} is always non-positive but $A'_1$ in
theorem \ref{thm:mainaltversionintro} may be positive. 
Furthermore, if $A'_1=0$ then we either have $P_1=0$ or $s_-=0$. If $P_1 >0$
and $s_-=0$ the best strategy consists in taking $q=s_-=0$ in
$H^\square_{\epsilon,q}$ to mininize the linear term. This choice excludes the
possibility of a negative quadratic bound (which is still possible if $P_1=0$
and $s_-=0$). This explains the appearance of $P_1$ instead of $A'_1$ in the
criterion. See the proofs for more details.

Both theorems \ref{thm:mainintro} and \ref{thm:mainaltversionintro} are a consequence of similar upper bounds for perturbations of periodic operators restricted to a fiber, see theorems \ref{thm:main} and \ref{thm:mainaltversion}. For these upper bounds, we present complementary lower bounds in lemmas \ref{lem:convinfspec} and \ref{lem:pos-convinfspec}.
\section{Periodic comparison operators}
In the present section we reduce the problem of studying $E_\epsilon$ to that of understanding 
certain adapted operators which are periodic with respect to a sublattice. 
Define 
$$\square_n := \bigcup_{\substack{m \in \gamma \\ |m|\leq nN}} \square + m$$
and $\chi_{n} := \chi_{\square_n}$, i.e.
$$
\chi_n(x) = 
\left\{
  \begin{array}{ll}
    1 & \quad \textrm{ for } x \in \square_n \\
    0 & \quad \textrm{ otherwise.}
  \end{array} 
\right.
$$ Note that $\square_0 = \square$ and that $\square_n$ is just the
collection of $(2n+1)^d$ disjoint translates of $\square$. Let us start
by stating the following lemma.
\begin{lem}
\label{lem:limitquasiperiodic}
Let $u$ be a $\theta$-quasi-$\gamma$-periodic function, i.e. such that for all $n\in \Z^d$ and $k \in \gamma$ we have
$$u(n+k) = e^{- i\theta \cdot k} u(n).$$ 
Define,
$$u_n := \chi_n u, \quad n=0,1,2,\ldots.$$
Then
$$\lim_{n \to \infty} \frac{\left\langle u_n , H_{q, \epsilon} u_n \right \rangle_{\ell^2(D)}}{\left\| u_n \right\|_{\ell^2(D)}^2} = \frac{\left\langle u_0, H_{q, \epsilon}^\square(\theta)u_0 \right\rangle_{\ell^2(\square)}}{\left\|  u_0 \right\|_{\ell^2(\square)}^2}.$$
\end{lem}
The proof of this lemma is found in the appendix.

For the definition of periodic comparison operators we introduce
\begin{equation}\label{def:omegaper}
  \Omega_\textrm{per}^n := \{ \omega \in \Omega : \omega \text{ is periodic w.r.t. } n\gamma\}.
\end{equation}
We now state the first comparison theorem.
\begin{thm}
\label{thm:discupperbound}
Assume \eqref{hypa}, \eqref{hypb} and either \eqref{hypc} or \eqref{hypcprime}. Let 
$\epsilon \geq 0$, $n \in \N$ and let $\omega \in \Omega_\textrm{per}^n $ be a $n\gamma$-periodic sequence of real numbers satisfying $\omega \in (S_\mu)^\gamma$, i.e. $\omega_k \in S_\mu$ for all $k \in \gamma$. Then, we have
$$\sigma(H_{\omega,\epsilon}) \subset \Sigma_\epsilon.$$
\end{thm}
We immediately deduce the following upper bound on
the minimum of the spectrum.
\begin{cor}
\label{cor:discupperbound}
Assume \eqref{hypa}, \eqref{hypb} and either \eqref{hypc} or \eqref{hypcprime}. Let $\epsilon \geq 0$, then
$$\bottom \leq \inf_{q \in S_\mu} E_{q,\epsilon}.$$
\end{cor}
\begin{proof}[Proof of theorem \ref{thm:discupperbound}]
 For the calculation below, we need a Weyl sequence of
compactly supported functions. This can indeed be done, since we only deal with bounded operators.  Fix $\omega \in \Omega_\textrm{per}^n $ and $E \in \sigma(H_{\omega,\epsilon})$. By
Floquet--Bloch theory, there exists some $\theta$ and some normalized state
$f\in \ell^2(\square)$ for which 
$$E = \left\langle \finhamper[\theta] f, f \right\rangle.$$
We extend $f$ as  a
$\theta$-quasi-$\gamma$-periodic function, i.e. for any $x \in \Z^d$ let $k \in \gamma$ such that $x-k \in \square$ and let
$$f(x):=e^{i \theta \cdot k}f(x-k).$$
Using lemma \ref{lem:limitquasiperiodic}, extract a sub-sequence $\{f_n\}$ from the sequence of functions $\left\{\displaystyle\frac{\chi_n f}{ \|\chi_n f \|_2} \right\}$, such that
$$\left| \left\langle (H_{\omega,\epsilon}-E) f_n , f_n \right\rangle \right| \leq 1/n$$
and satisfying, for a sequence $l_n \in \N$,
$$\supp f_n \subset \Lambda_{l_n},$$
where $\Lambda_{l_n}$ is a cube centered at zero and sidelength $l_n$.
For $x \in \gamma$ we define
$$\Omega(x,n) : = \left\{  \omega' \in \Omega : \forall k \in (x + \Lambda_{l_n}) \cap \gamma: \left| \epsilon (\omega'_k - \omega_k) \right| \leq 1/n \right\}.$$
Now, since $\omega \in (S_\mu)^\gamma$,
$$\Pro \left[ \Omega(x,n) \right] > 0,$$
and for $x,y \in \gamma$ satisfying $|x-y| > l_n$, the events $\Omega(x,n)$ and $\Omega(y,n)$ are independent (and identically distributed).
Using Borel--Cantelli lemma, we see that the event
$$\Omega' := \bigcap_{n \in \N} \bigcup_{x \in \gamma} \Omega(x,n)$$
has probability one. 

From the definition of $\Omega(x,n)$, we have that given $\omega' \in
\Omega'$ and $n \in \N$, there exists a $x(n,\omega')$ such that
$\omega' \in \Omega(x,n)$.  We write from now on $\tau_{x(n,\omega')}
f_n$ for the translated function $f_n( \cdot - x(n,\omega'))$. Let
$\omega' \in \Omega'$ and $n \in \N$, and calculate
\begin{align*}
 &\left\langle \vphantom{V_\omega} (H_{\omega',\epsilon} - E) \tau_{x(n,\omega')} f_n,\tau_{x(n,\omega')} f_n \right\rangle \\
=& \left\langle\vphantom{V_\omega} (H_0 - E) \tau_{x(n,\omega')} f_n, \tau_{x(n,\omega')} f_n \right\rangle + \epsilon \left\langle\vphantom{V_\omega} V_{\omega'}  \tau_{x(n,\omega')} f_n, \tau_{x(n,\omega')} f_n\right\rangle \\
=& \left\langle\vphantom{V_\omega} (H_0 - E) f_n,f_n \right\rangle + \epsilon \left\langle V_{\omega} \tau_{x(n,\omega')} f_n,\tau_{x(n,\omega')} f_n \right\rangle \\
& \phantom{ \,\left\langle\vphantom{V_\omega} (H_0 - E) f_n,f_n \right\rangle} \mathbin{+} \epsilon \left\langle V_{\omega'-\omega} \tau_{x(n,\omega')} f_n, \tau_{x(n,\omega')} f_n \right\rangle \\
=& \left\langle (H_{\omega,\epsilon} - E) f_n, f_n \right\rangle + \epsilon \left\langle V_{ \omega'-\omega} \tau_{x(n,\omega')} f_n, \tau_{x(n,\omega')} f_n\right\rangle
\end{align*}
Note that $| \epsilon V_{ \omega'-\omega} (x)| \leq \|V^\square\| /n$ 
if $x \in \supp f_n( \cdot - x(n,\omega'))$, so that
\begin{align}
\label{eq:appropriate-eigenfunction}
\left| \left\langle (H_{ \omega',\epsilon} - E) \tau_{x(n,\omega')} f_n, \tau_{x(n,\omega')} f_n \right\rangle - \left\langle (H_{\omega,\epsilon} - E) f_n, f_n \right\rangle \right| & \leq \frac{1}{n} \|V^\square\|.
\end{align}
Here $\|V^\square\|$ denotes the operator norm. 
In the particular case that $V$ is a multiplication operator
it coincides with the supremum $\|V^\square\|_{\infty}$ of the function $V$.
Inequality \eqref{eq:appropriate-eigenfunction} 
implies that $\tau_{x(n,\omega')} f_n$ is a Weyl sequence.
\end{proof}
\begin{rem}
  This is an adaptation of a well known argument of Kirsch and
  Martinelli \cite{kirschmartinelli1982almostsurespectrum} in the
  continuous setting, with $S_\mu$ connected and $V^\square$ a
  multiplication operator.
\end{rem}
\begin{rem}\label{rem:overlap}
 When the random potential is diagonal (as in the introduction), the proof above can be adapted to \emph{overlapping}, but compactly supported single-site  potentials $V^\square \in \ell^\infty(D)$ as long as
\begin{equation*}
  \sum_{n \in \gamma} V^\square(\cdot - n) \not \equiv 0.
\end{equation*}
Note that if this condition does not hold then $H_q=H_0$ for all $q$. One way around this problem would be to consider periodic (non-constant) sequences of coupling constants $\omega_n$ such that the resulting periodic potential is not zero.
\end{rem}

To prove the following converse to theorem \ref{thm:discupperbound} we define 
\
\begin{equation*}
  \Omega_\textrm{per} := \{ \omega \in \Omega : \exists n \in \N \text{ such that } \omega \text{ is periodic w.r.t. } n\gamma\} = \bigcup_{n \in \N} \Omega_\textrm{per}^n.
\end{equation*}
\begin{lem}
\label{thm:disclowerbound}
Denote by $\Sigma_\epsilon$ the almost sure spectrum of $H_{\omega,\epsilon}$. Then:
\
\begin{equation*}
  \Sigma_\epsilon \subset \overline{\bigcup_{\omega \in \Omega_\text{per}} \sigma(H_{\omega,\epsilon})}.
\end{equation*}
\end{lem}
\begin{proof}
Let $n \in \N$ and set
\
\begin{align*}
  \omega^{(n)}_k = \omega_k & \quad \text{ for } f \in \square_N \\
  \omega_k^{(n)} = \omega_j & \quad \text{ if } j-k \in N\gamma.
\end{align*}
Let $C_0(\Z^d)$ be the set of compactly supported functions in $\ell^2(\Z^d)$. Choose any $\varphi \in C_0(\Z^d)$. Then
\
\begin{equation*}
  \lim_{n \to \infty} \| H_{\omega,\epsilon} \varphi -  H_{\omega^{(n)},\epsilon} \varphi \| = 0,
\end{equation*}
i.e. we have strong convergence $H_{\omega^{(n)}} \to H_\omega$. Since the operators $H_\omega$ are bounded, the set $C_0$ is an operator core for $H_\omega$. This implies that we have strong convergence on the whole $\ell^2(\Z^d)$. 

By the resolvent equation, for any $E \in  \R \diagdown \Sigma$,
\begin{align*}
  & (H_{\omega,\epsilon} -E)^{-1} -  (H_{\omega^{(n)},\epsilon} -E)^{-1} \\
  = & (H_{\omega,\epsilon} -E)^{-1} ( V_\omega - V_{\omega^{(n)}}) (H_{\omega^{(n)},\epsilon} -E)^{-1} \\
 = & (H_{\omega^{(n)},\epsilon} -E)^{-1} ( V_\omega - V_{\omega^{(n)}}) (H_{\omega,\epsilon} -E)^{-1},
\end{align*}
which converges strongly to $0$.
We know that if $E \in  \R \diagdown \Sigma$, then $(H_{\omega,\epsilon} - E)^{-1}\varphi \in \ell^2(D)$ for any $\varphi \in \ell^2(D)$ and that, using theorem \ref{thm:discupperbound}, the inclusion $\sigma(H_{\omega,\epsilon}(n)) \subset \Sigma$ holds for any $\omega$ in the support of the product measure $\bigotimes\limits_{D} \mu$. To conclude, we apply theorem VIII.24 in \cite{reedsimon01} which tells us that
\
\begin{equation*}
  \sigma(H_{\omega,\epsilon}) \subset \overline{ \bigcup_{n \in \N} \sigma(H_{\omega^{(n)}}) } .
\end{equation*}
This finishes the proof.
\end{proof}
In particular we obtain the following corollary.
\
\begin{cor}
As before we set
  \begin{equation*}
    \Omega_\textrm{per} := \{ \omega \in \Omega : \exists n \in \N \text{ such that } \omega \text{ is periodic w.r.t. } n\gamma\}
  \end{equation*}
and denote by $\Sigma_\epsilon$ the almost sure spectrum of $H_{\omega,\epsilon}$. Then:
\
\begin{equation*}
 \inf \Sigma_\epsilon = \inf \bigcup_{\omega \in \Omega_\text{per}}  \sigma(H_{\omega,\epsilon}).
\end{equation*}
\end{cor}

\section{Perturbation calculation}
\label{sec:perturbationcalculation}
For the readers convenience we recall the definition of the constants
$A_1$ and $A_2$, the notation and the statement of the theorems before the
proofs.  By the continuity of the Floquet--Bloch
eigenvalues there exists some $\theta$ such that
$$E_0:=\inf \sigma(H_0) = \inf \sigma(H_0^\square(\theta)) = 0.$$ We denote by
 $\Theta \subset \square^*$ the compact set of $\theta$ for which the
last equality holds. From now on we fix some $\theta \in \Theta$, so
the quantities below will depend on $\theta$. Let $\mathcal V_0 $ be
the eigenspace of $H_0(\theta)$ associated to the eigenvalue $E_0 =
0$, $p$ its multiplicity and choose an orthonormal basis $\psi_j$,
$j=1,\ldots,p$ spanning $\mathcal V_0$ and
diagonalizing the Hermitian matrix $A \in \C^{p \times p}$, given by the coefficients
$$A_{ij} := \left\langle \psi_i, V^\square\psi_j \right\rangle.$$
We take the eigenvalues of the matrix $A$ in the ascending order counting
multiplicities so that $P_1:=A_{11} =
\left\langle \psi_1, V^\square\psi_1 \right\rangle $ is the minimal
eigenvalue and $P_p:=A_{pp} = \left\langle \psi_p, V^\square\psi_p
\right\rangle $ is the maximal eigenvalue of $A$.  

\subsection{Sign-changing random variables}
In this subsection we assume \eqref{hypc} to hold. We will only treat this case in detail as the calculation for positive random variables is very similar.  
Recall from \eqref{hypc} that $s_- < 0 < s_+$.   
We define the following quantities :
\begin{equation} 
  A_1:= \inf_{q \in S_\mu} \inf_{\substack{\psi \in \mathcal V_0
    \\ \|\psi\|_{\ell^2(\square)} = 1}} q \left\langle \psi,
  V^\square \psi \right\rangle = \min( s_+ P_1, s_- P_p) \leq 0,
\end{equation}
and
\begin{equation} 
A_2:= - \max(s_-^2,s_+^2) \sup_{\substack{\psi \in \mathcal V_0
    \\ \|\psi\|_{\ell^2(\square)} = 1}} \sup_{\substack{\varphi \in \mathcal V_0^\perp
    \\ \|\varphi\|_{\ell^2(\square)} = 1}} \frac{\left| \left\langle \psi,
  V^\square \varphi \right\rangle \right|^2}{\left\langle H^\square_0(\theta)
  \varphi, \varphi \right\rangle} \leq 0.
\end{equation}
Note that the sign of $A_1$ and $A_2$ is fixed.
We will prove the following theorem, which is only a restatement of theorem \ref{thm:mainintro}.
\begin{thm}
\label{thm:main}
Assume \eqref{hypa}, \eqref{hypb} and \eqref{hypc}. Fix $\theta \in \Theta$. 
Then, for $\epsilon>0$ small enough, if $A_1 \neq 0$,
\begin{equation}
  \bottom \leq \epsilon A_1, \nonumber
\end{equation}
whereas if $A_1 = 0$, but $A_2\neq0$, then
\begin{equation}
  \bottom \leq \epsilon^2 A_2 + O(\epsilon^3). \nonumber
\end{equation}
Finally, if $A_1 = A_2 = 0$, then
$$\bottom \leq 0.$$
\end{thm}
\begin{rem}
  \
  \begin{itemize}
  \item We remind that we have fixed $\theta$ to simplify notations, but
    $A_1$ and $A_2$ depend on $\theta$. The best bound for the behavior
    of the bottom of the spectrum is obtained by looking at each $\theta
    \in \Theta$ and taking the minimum.
    
  \item We see that our bound on the bottom of the spectrum behaves
    linearly, quadratically or it doesn't move with $\epsilon$. In the
    analogous setting in continuum space, if the unique continuation principle is
    not violated, then the analogous result does not allow the third
    case $A_1=A_2=0$. This leaves only the cases of a linear or a quadratic bound.
    
  \item The definition of the quantities $A_1$, $A_2$ may seem
    complicated at first sight, but these choices are optimal, in the
    sense of lemma \ref{lem:convinfspec} below, which is a converse of
    lemma \ref{lem:infspec} in the regime $\epsilon \ll 1$.
\end{itemize}
\end{rem}
Before proving the theorem, let us provide a
  much simpler, non-optimal upper bound for $A_2$ as well as a
  condition ensuring that $|A_1| + |A_2| \neq 0$.
\subsection{A simple non-degeneracy condition}
Theorem \ref{thm:main} tells us that if $A_2 \neq 0$, then the
expansion of the bottom of the spectrum is at least quadratic, but
if $A_1=A_2=0$, we can only say that the spectrum starts at zero. When
$V^\square$ is diagonal this only happens if the support of the
single-cell potential and the eigenfunctions $\psi_1, \ldots, \psi_p$
are disjoint (the $\psi_i$ were defined at the beginning of this section). 

Note, that in the continuous configuration space this can
only happen if the potential violates the unique continuation
principle. For a discussion on the validity of the unique continuation
principle see for instance \cite{wolff1993ucp}.

Let us discuss the condition in our general setting. 
First let us remark that  if $A_1=0$, then the matrix $A  \in  \C^{p \times p}$ vanishes identically, i.e. \
\begin{equation}\label{A1equalzero}
  A_1=0 \implies \sup_{\substack{\psi \in \mathcal V_0
    \\ \|\psi\|_{\ell^2(\square)} = 1}} \left| \left\langle \psi,
  V^\square \psi \right\rangle \right| = 0.
\end{equation}
and thus 
\begin{equation*}
(\forall \psi \in \mathcal V_0) \quad V^\square \psi \in \mathcal V_0^\perp.
\end{equation*}
The operator $H_0^\square$ is invertible on $\mathcal V_0^\perp$ and thus there exists some $\varphi \in \mathcal V_0^\perp$ such that
\begin{equation}
\label{eq:problem}
H_0^\square(\theta) \varphi =  V^\square  \psi^*.
\end{equation}
Hence, we have that
$$\left\langle V^\square \varphi, \psi^* \right\rangle = \left\langle H^\square_0(\theta) \varphi, \varphi \right\rangle.$$
Now, assume there exists some $\psi^* \in \mathcal V_0$ such that
  \begin{equation}
    \label{cond:nonzero}
    V^\square \psi^* \neq 0.
  \end{equation}
Then $\varphi$ in \eqref{eq:problem} does not vanish  and
$$A_2 \leq - \max(s_-^2,s_+^2)  \frac{\left|\left\langle V^\square \varphi, \psi^* \right\rangle\right|^2} {\left\langle H^\square_0(\theta) \varphi, \varphi \right\rangle} = - \max(s_-^2,s_+^2)  {\left\langle H^\square_0(\theta) \varphi, \varphi \right\rangle} <0,$$
because $\varphi \not \in \ker H_0^\square(\theta)$.
\begin{rem}
Formally, we have 
$$A_2 \leq - \max(s_-^2,s_+^2)  \left\langle \psi^*, V^\square H^\square_0(\theta)^{-1} V^\square \psi^* \right\rangle$$
when $A_1 = 0$.
\end{rem}
In the converse direction, $A_1=0$ together with $A_2=0$ implies that
\
\begin{equation*}
  (\forall \psi \in \mathcal V_0  \text{ and } \forall \varphi \in \mathcal \ell^2(\square)) \quad \bigl \langle V^\square \psi, \varphi \bigr\rangle = 0,
\end{equation*}
i.e. that 
$$(\forall \psi \in \mathcal V_0) \quad V^\square \psi =0.$$

We summarize the above discussion as follows.
\begin{lem} \label{cor:simple}
  Under the assumptions of theorem \ref{thm:main} we have that
\
\begin{equation*}
  A_1 = 0 \textrm{ and } A_2 = 0 \textrm{\quad if and only if \quad} (\forall \psi^* \in \mathcal V_0) ~  V^\square \psi^* = 0.
\end{equation*}
\end{lem}

\subsection{Proof of theorem \ref{thm:main}}
We subdivide the proof of theorem \ref{thm:main} into two lemmas.
The first covers both types of sign assumptions on the random variables.
\begin{lem}\label{lem:upperboundperiodic}
Assume \eqref{hypa}, \eqref{hypb}, and either \eqref{hypc} or \eqref{hypcprime}. Let $u \in  \ell^2(\square)$ and $\bottom$ as in \eqref{def:bottomspec}. Then,
$$\bottom \leq \inf_{q \in S_\mu} \inf_{u\in \ell^2(\square)} \frac{\left \langle H_{q, \epsilon}^\square(\theta) u,u \right \rangle}{\|u\|_{\ell^2(\square)}} 
\qquad \text{ for any } \theta \in \square^* $$
\end{lem}
 \debdem 
By Corollary \ref{cor:discupperbound} it is enough to consider the periodic
realizations of the potential.  By the Courant--Weyl--Fischer min--max
principle,
\begin{equation}
\label{eq:minmaxprinciple}
\bottom \leq E_{q,\epsilon} = \min \sigma(H_{q, \epsilon}) = \inf_{\substack{a \in \ell^2(\Z^d) \\ \|a\|_2=1}} \left \langle 
H_{q, \epsilon} a,a \right \rangle.
\end{equation}
Finally, by lemma \ref{lem:limitquasiperiodic},
\begin{equation}
 \inf_{\substack{a \in \ell^2(\Z^d) \\ \|a\|=1}} \left \langle 
H_{q, \epsilon} a,a \right \rangle \leq \inf_{u\in \ell^2(\square)} \frac{\left \langle H_{q, \epsilon}^\square(\theta) u,u \right \rangle}{\|u\|_{\ell^2(\square)}}.
\end{equation}
This proves the lemma.
\findem

We state now the second lemma.
It applies to the case of sign-changing random variables.
\begin{lem} \label{lem:infspec}
  Let $A_1$ and $A_2$ as in \eqref{def:A1}, \eqref{def:A2}, assume
  \eqref{hypa}, \eqref{hypb}, and \eqref{hypc} and fix $\theta \in \Theta$.  Then, for $\epsilon>0$ small enough, if $A_1 \neq 0$,
\begin{equation*}
  \inf_{q \in S_\mu} \inf_{\|u\|_{\ell^2(\square)} = 1} \left \langle H_{q, \epsilon}^\square(\theta) u,u \right \rangle \leq \epsilon A_1,
\end{equation*}
whereas if $A_1 = 0$, but $A_2\neq 0$, then
\begin{equation*}
  \inf_{q \in S_\mu} \inf_{\|u\|_{\ell^2(\square)} = 1} \left \langle H_{q, \epsilon}^\square(\theta) u,u \right \rangle \leq \epsilon^2 A_2 + O(\epsilon^3)
\end{equation*}
Finally, if $A_1 = A_2 = 0$, then
\begin{equation*}
  \inf_{q \in S_\mu} \inf_{\|u\|_{\ell^2(\square)} = 1} \left \langle H_{q, \epsilon}^\square(\theta) u,u \right \rangle \leq 0.
\end{equation*}
\end{lem}
\debdem
It is enough to show that for some $q \in S_\mu$, there is some normalized state $u \in \ell^2(\square)$ satisfying
$$\left \langle H_{q, \epsilon}^\square(\theta) u,u \right \rangle \leq \epsilon A_1 \quad \text{or} \quad \epsilon^2 A_2 + O(\epsilon^3) \quad \text{or} \quad 0 \quad \text{resp.}.$$

Let $\psi \in \mathcal V_0$ and $\varphi \in \mathcal V_0^\perp$, to be chosen later, and $u
= \psi + \epsilon q \varphi$. We assume furthermore $\|\psi\| =1$. We expand
$$\|u\|^2 = \|\psi\|^2 + \epsilon^2 q^2  \| \varphi \|^2$$
and thus
\begin{equation}
\label{eq:devinvnorm}
1/\| u \|^2 = 1 - \epsilon^2 q^2 \| \varphi \|^2 + O(\epsilon^4 \|\varphi\|^4).
\end{equation}

We calculate  the kinetic
energy of this state, i.e.
\begin{equation}
\label{eq:devkine}
  \left\langle H_0^\square(\theta) u,u \right\rangle = \left\langle H_0^\square(\theta) \psi,\psi \right\rangle + 2 \epsilon q \Re \left\langle H_0^\square(\theta) \psi, \varphi \right\rangle  + \epsilon^2 q^2 \left\langle H_0^\square(\theta) \varphi, \varphi \right\rangle.
\end{equation}
Because $\psi \in \mathcal V_0$ and $E_0 = 0$, we see that \eqref{eq:devkine} becomes
\begin{equation*}
  \left\langle H_0^\square(\theta) u,u \right\rangle = \epsilon^2 q^2 \left\langle H_0^\square(\theta) \varphi, \varphi \right\rangle.
\end{equation*}
We expand the potential energy as
$$\epsilon q \left\langle V^\square u,u \right\rangle = \epsilon q \left\langle V^\square \psi, \psi \right\rangle + 2 \epsilon^2 q^2 \Re \left\langle V^\square \varphi, \psi \right\rangle  + \epsilon^3 q^3 \left\langle V^\square \varphi, \varphi \right\rangle.$$
Thus,
\begin{align}
  \left\langle H^\square_{\epsilon, q}(\theta) u, u \right\rangle = & \epsilon q \left\langle V^\square \psi, \psi \right\rangle  + \epsilon^2 q^2
 \Bigl( \left\langle H^\square_{0}(\theta) \varphi, \varphi \right\rangle  + 2 \Re
\left\langle \psi, V^\square \varphi \right \rangle \Bigl) \nonumber\\
& + \epsilon^3 q^3
\left\langle V^\square \varphi, \varphi\right\rangle.
\label{eq:expansion}
\end{align}

\subsubsection*{Case $A_1 \neq 0$.}
Note that in this case $P_1P_p \neq 0$.
From now on we assume that $s_+ P_1 \leq s_- P_p$. If this is not the
case, we can always replace $V^\square \mapsto - V^\square$ and
$\omega_n \mapsto -\omega_n$ to get an equivalent model. In this case,
we take $\psi =\psi_1$, $\varphi =0$ and $q=s_+$. Then,
\eqref{eq:expansion} becomes
$$ \left\langle H^\square_{\epsilon, q}(\theta) u, u \right\rangle = \epsilon q \left\langle V^\square \psi_1, \psi_1 \right\rangle = \epsilon s_+ P_1,$$
which proves the result in this case, as $u$ is normalized.

\subsubsection*{Case $A_1 = 0$ and $A_2 \neq 0$.}
First let us remark that  if $A_1=0$ then the matrix $A  \in  \C^{p \times p}$ vanishes identically, i.e. \
\begin{equation}
  A_1=0 \implies \sup_{\substack{\psi \in \mathcal V_0
    \\ \|\psi\|_{\ell^2(\square)} = 1}} \left| \left\langle \psi,
  V^\square \psi \right\rangle \right| = 0.
\end{equation}

In this case we have that, for any $\psi \in \mathcal V_0$ and
$\varphi \in \mathcal V_0^\perp$, the expansion \eqref{eq:expansion}
becomes
\begin{equation}
  \label{eq:expansionquadratic}
  \left\langle H^\square_{\epsilon, q}(\theta) u, u \right\rangle =  \epsilon^2 q^2
  \Bigl( \left\langle H^\square_{0}(\theta) \varphi, \varphi \right\rangle  + 2 \Re
  \left\langle \psi, V^\square \varphi \right \rangle \Bigl) + \epsilon^3 q^3
  \left\langle V^\square \varphi, \varphi\right\rangle .
\end{equation}
Note that, for $\psi \in \mathcal V_0$ and $\varphi \in \mathcal V_0^\perp$ such that
$$\|\psi\|_2 = \|\varphi\|_2 = 1$$
the map
\begin{align*}
  (\varphi,\psi) & \mapsto \frac{\left| \left\langle \psi,
    V^\square \varphi \right\rangle \right|^2}{\left\langle H^\square_0(\theta)
    \varphi, \varphi \right\rangle}
\end{align*}
is continuous. Given that the spaces involved are finite-dimensional
and their respective unit balls thus compact, we know that there
exists a couple $(\psi^*, \varphi^*)$ maximizing this quantity, i.e.
$$\frac{\left| \left\langle \psi^*, V^\square \varphi^*
  \right\rangle\right|^2}{\left\langle H^\square_0(\theta) \varphi^*,
  \varphi^* \right\rangle} = \sup_{\substack{\psi \in \mathcal V_0
    \\ \|\psi\|_{\ell^2(\square)} = 1}} \sup_{\substack{\varphi \in
    \mathcal V_0^\perp \\ \|\varphi\|_{\ell^2(\square)} = 1}}
\frac{\left| \left\langle \psi, V^\square \varphi \right\rangle\right|
  ^2}{\left\langle H^\square_0(\theta) \varphi, \varphi
  \right\rangle}.$$ 
Let $\psi = \psi^*$ and $\varphi=\lambda
\varphi^*$ in the definition
of $u$, where
$$\lambda = -\frac{\left\langle \psi^*, V^\square \varphi^* \right \rangle}{\left\langle H^\square_{0}(\theta) \varphi^*, \varphi^* \right\rangle} \in \C.$$
Replacing, we see that
\begin{align*}
 {\left\langle H^\square_{0}(\theta) \varphi, \varphi \right\rangle} +
 2 {\Re \left\langle \psi, V^\square \varphi \right \rangle} 
= & |\lambda|^2 {\left\langle H^\square_{0}(\theta) \varphi^*, \varphi^* \right\rangle} +  2 {\Re \overline \lambda \left\langle \psi^*, V^\square \varphi^* \right \rangle}
\\  = & \frac{\left| \left\langle \psi^*, V^\square \varphi^* \right
   \rangle \right|^2}{\left\langle H^\square_{0}(\theta) \varphi^*, \varphi^*
   \right\rangle} - 2 \frac{\left| \left\langle \psi^*, V^\square \varphi^* \right
   \rangle \right|^2}{\left\langle H^\square_{0}(\theta) \varphi^*, \varphi^*
   \right\rangle}
\\ = & -
 \frac{\left| \left\langle \psi^*, V^\square \varphi^* \right
   \rangle \right|^2}{\left\langle H^\square_{0}(\theta) \varphi^*, \varphi^*
   \right\rangle}. \label{minimizeme}
\end{align*}
Using this in \eqref{eq:expansionquadratic} and letting $q^2=\max(s_-^2,s_+^2)$, we obtain
\begin{align*}
  \left\langle H^\square_{\epsilon, q}(\theta) u, u \right\rangle & =  - \epsilon^2 \max(s_-^2,s_+^2) \frac{\left| \left\langle \psi^*, V^\square \varphi^* \right
   \rangle \right|^2}{\left\langle H^\square_{0}(\theta) \varphi^*, \varphi^*
   \right\rangle} + O(\epsilon^3 q^3 \|\varphi^*\|^2 ) 
  \\ & = \phantom{-} \epsilon^2 \, A_2 + O(\epsilon^3 q^3 \|\varphi\|^2 ) .
\end{align*}
Normalizing $u$ by multiplying by \eqref{eq:devinvnorm} gives the result.

\subsubsection*{Case $A_1 = 0$ and $A_2 = 0$.}
Choose $\varphi = 0$ and any normalized $\psi \in \mathcal V_0$. The development using $u$ in this case gives
$$\left\langle H^\square_{\epsilon, q}(\theta) u, u \right\rangle =  \epsilon^3 q^3
\left\langle V^\square \varphi, \varphi\right\rangle = 0$$ 
and this yields the desired result.
\findem 

We prove the following converse lemma.
\begin{lem} \label{lem:convinfspec}
  Let $A_1$ and $A_2$ as in \eqref{def:A1}, \eqref{def:A2}, assume
  \eqref{hypa}, \eqref{hypb} and \eqref{hypc}, and fix $\theta \in \Theta$.  Then, for $\epsilon>0$ small enough, if $A_1 \neq 0$,
\begin{equation*}
  \inf_{q \in S_\mu} \inf_{\|u\|_{\ell^2(\square)} = 1} \left \langle H_{q, \epsilon}^\square(\theta) u,u \right \rangle \geq \epsilon A_1 + O(\epsilon^{3/2}),
\end{equation*}
whereas if $A_1 = 0$, then
\begin{equation*}
  \inf_{q \in S_\mu} \inf_{\|u\|_{\ell^2(\square)} = 1} \left \langle H_{q, \epsilon}^\square(\theta) u,u \right \rangle \geq \epsilon^2 A_2 + O(\epsilon^3)
\end{equation*}
Finally, if $A_1 = A_2 = 0$, then
\begin{equation*}
  \inf_{q \in S_\mu} \inf_{\|u\|_{\ell^2(\square)} = 1} \left \langle H_{q, \epsilon}^\square(\theta) u,u \right \rangle \geq 0.
\end{equation*}
\end{lem}
\debdem
Fix $\epsilon >0$ and let $q_\epsilon \in S_\mu$ be a value which minimizes the map
\begin{equation}
\label{eq:infimum}
  q \mapsto \inf_{\|u\|_{\ell^2(\square)} = 1} \left \langle H_{\epsilon,
  q}^\square(\theta) u,u \right \rangle.
\end{equation}
We don't now much about $q_\epsilon$, but we know a-priori
$q_\epsilon \in [s_-,s_+]$. This is the only property we will use of $q_\epsilon$.
For simplicity, we write in the sequel simply $q$ for $q_\epsilon$.
We lower bound the right hand side of (\ref{eq:infimum})
by minimizing over a larger set by writing
\
\begin{equation*}
\inf_{\|u\|_{\ell^2(\square)} = 1} \left \langle H_{\epsilon,
  q}^\square(\theta) u,u \right \rangle \geq \inf_{\substack{\psi \in
    \mathcal V_0 \\ \|\psi\|_{\ell^2(\square)} \leq 1}}
\inf_{\substack{\varphi \in \mathcal V_0^\perp
    \\ \|\varphi\|_{\ell^2(\square)} \leq 1}} \left \langle H_{\epsilon,
  q}^\square(\theta) (\psi+\varphi),(\psi+\varphi) \right \rangle. 
\end{equation*}
By continuity and compactness, 
there exists some pair $(\psi^*,\varphi^*)=(\psi^*_\epsilon,\varphi^*_\epsilon) $ 
in $\mathcal V_0 \times \mathcal V_0^\perp$ 
realizing the infimum on the right hand side. We see that
\begin{equation} \label{eq:normvarphi}
  \left\langle H^\square_0(\theta)
  (\psi^* + \varphi^*), (\psi^* + \varphi^*) \right\rangle = \left\langle H^\square_0(\theta)
  \varphi^*, \varphi^* \right\rangle \geq g \|\varphi^*\|_{\ell^2(\square)}^2,
\end{equation}
where the constant $g$ is the spectral gap of $H^\square_0$. Due to our normalization $g$ coincides with the (positive) second eigenvalue of $H^\square_0$.
We study the different cases.
\subsubsection*{Case $A_1 \neq 0$.}
 From lemma \ref{lem:infspec}, we know already that 
\begin{equation} \label{eq:negspec}
  |A_1|+|A_2| \neq 0 \implies \left \langle H_{q, \epsilon}^\square(\theta) (\psi^*+\varphi^*),(\psi^*+\varphi^*) \right \rangle < 0.
\end{equation}
Using \eqref{eq:normvarphi} and \eqref{eq:negspec} we get that
\begin{equation*}
   \|\varphi^*\|^2_{\ell^2(\square)} 
\leq - g^{-1} \epsilon q \left\langle V^\square (\psi^* + \varphi^*), (\psi^* + \varphi^*) \right\rangle 
\leq 4g^{-1} \|V^\square\| \epsilon q,
\end{equation*}
where $\|V^\square\|$ is the operator norm of $V^\square$.
We deduce then that
\begin{multline*}
  \inf_{\|u\|_{\ell^2(\square)} = 1} \left \langle H_{q, \epsilon}^\square(\theta) u,u \right \rangle \\
 \geq  \epsilon q  \left\langle V^\square \psi^*, \psi^* \right\rangle + 2 \epsilon q \Re \left\langle V^\square \varphi^*, \psi^* \right\rangle   + \epsilon q \left\langle V^\square \varphi^*, \varphi^* \right\rangle \\
  \geq  \epsilon A_1 - 4  g^{-1/2}\epsilon^{3/2} q^{3/2} \|V^\square\|^{3/2}_\infty - 4 g^{-1}\epsilon^2q^2 \|V^\square\|^2_\infty.
\end{multline*}
\subsubsection*{Case $A_1 = 0$ and $A_2 \neq 0$.}
In this case, due to (\ref{A1equalzero}),
\begin{multline*} 
    \left \langle H_{q, \epsilon}^\square(\theta) (\psi^*+\varphi^*),(\psi^*+\varphi^*) \right\rangle \\
=  \left\langle H^\square_0(\theta)  \varphi^*, \varphi^* \right\rangle + 2 \epsilon q \Re \left\langle V^\square \varphi^*, \psi^* \right\rangle  + \epsilon q \left\langle V^\square \varphi^*, \varphi^* \right\rangle.
\end{multline*}
Using \eqref{eq:negspec} we see that $\varphi^* \neq 0$. Furthermore, \eqref{eq:normvarphi} and \eqref{eq:negspec} together imply that
\begin{align*}
   \|\varphi^*\|_{\ell^2(\square)}^2 
\leq& \epsilon q g^{-1} \|V^\square\|  ( 2  \|\varphi^*\|_{\ell^2(\square)} +  \|\varphi^*\|_{\ell^2(\square)}^2) 
\\ \leq& 3 \epsilon q g^{-1} \|V^\square\|  \|\varphi^*\|_{\ell^2(\square)}.
\end{align*}
Note that $ \|\varphi^*\|_{\ell^2(\square)}$ is on both sides of the inequality. Simplifying,
\begin{equation} \label{ineq:smallnorm}
   \|\varphi^*\|_{\ell^2(\square)} \leq 3 \epsilon q g^{-1} \|V^\square\| .
\end{equation}
Expanding as $\epsilon \to 0$, employing \eqref{ineq:smallnorm} and then simply multiplying by $1=|\lambda|^2/|\lambda|^2=\overline \lambda/\overline \lambda$, we write 
\begin{multline*} 
  \inf_{\|u\|_{\ell^2(\square)}=1} \left \langle H_{q, \epsilon}^\square(\theta) u,u \right \rangle \\
\geq   \left \langle H_0^\square(\theta) \varphi^*, \varphi^* \right\rangle  +  2 \Re \epsilon q  \left\langle V^\square \psi^*, \varphi^* \right\rangle +   \epsilon q \left\langle V^\square \varphi^*, \varphi^* \right\rangle  \\
=   |\lambda|^2 \frac{\left\langle H_0^\square(\theta) \varphi^*, \varphi^* \right\rangle}{|\lambda|^2}  +  2 \Re  \overline \lambda  \epsilon q   \frac{\left\langle V^\square \psi^*, \varphi^* \right\rangle}{\overline \lambda} +   O(\epsilon^3)
\end{multline*}
We choose $\lambda$ as 
\
\begin{equation*}
  \lambda = - \frac{\left\langle H_0^\square(\theta) \varphi^*, \varphi^* \right\rangle}{\left\langle V^\square \psi^*, \varphi^* \right\rangle}.
\end{equation*} 
We will show that $\lambda$ is well defined for small $\epsilon$. Indeed, using \eqref{eq:normvarphi} and \eqref{eq:negspec} we see that
\
\begin{align}\nonumber
  - 2 \epsilon q \Re \left\langle V^\square \varphi^*, \psi^* \right\rangle  & \geq \left\langle H^\square_0(\theta)  \varphi^*, \varphi^* \right\rangle +  \epsilon q \left\langle V^\square \varphi^*, \varphi^* \right\rangle \\
  & \geq  g \|\varphi^*\|_{\ell^2(\square)}^2 - \epsilon q  \|V^\square\| \|\varphi^*\|_{\ell^2(\square)}^2.
\label{eq:lower-bound}
\end{align}
Since we know that $\varphi^* \neq 0$ the lower bound in (\ref{eq:lower-bound}) is strictly positive for sufficiently 
small $\epsilon $. We conclude that $\lambda$ is well defined (and different from $0$) for $\epsilon$ small enough.

Using our choice of $\lambda$ gives
\begin{multline*}
 |\lambda|^2 \frac{\left\langle H_0^\square(\theta) \varphi^*, \varphi^* \right\rangle}{|\lambda|^2}  +  2 \Re  \overline \lambda  \epsilon q   \frac{\left\langle V^\square \psi^*, \varphi^* \right\rangle}{\overline \lambda} +   O(\epsilon^3) \\
=(|\lambda|^2  -  2 \Re \overline \lambda \epsilon q)  \frac{\left| \left\langle V^\square \psi^*, \varphi^* \right\rangle \right|^2}{\left\langle H_0^\square(\theta) \varphi^*, \varphi^* \right\rangle} +   O(\epsilon^3) \\
\end{multline*}
To bound the last expression from below, we use the trivial bound
$|\lambda|^2 - 2 \Re \overline \lambda \epsilon q 
\geq |\lambda|^2 - 2 |\lambda| \epsilon q \geq -\epsilon^2 q^2$ 
as well as $ - \frac{q^2}{\max(s_-^2,s_+^2)} \geq -1$, and obtain
\begin{multline*}
(|\lambda|^2 - 2\Re \overline \lambda \epsilon q)  
\frac{\left| \left\langle V^\square \psi^*, \varphi^* \right\rangle \right|^2}{\left\langle H_0^\square(\theta) \varphi^*, \varphi^* \right\rangle} +   O(\epsilon^3) \\
\geq
-\epsilon^2 q^2 \frac{\left| \left\langle V^\square \psi^*, \varphi^* \right\rangle \right|^2}{\left\langle H_0^\square(\theta) \varphi^*, \varphi^* \right\rangle} +   O(\epsilon^3) 
\geq
\frac{A_2 \epsilon^2 q^2}{\max(s_-^2,s_+^2)} +   O(\epsilon^3) \\
 \geq  \epsilon^2  A_2 +  O(\epsilon^3),
\end{multline*}

\subsubsection*{Case $A_1 = 0$ and $A_2 = 0$.}

In this case
\begin{align} 
 \inf_{\|u\|_{\ell^2(\square)}} \left \langle H_{q, \epsilon}^\square(\theta) u,u \right \rangle &=    \left \langle H_0^\square(\theta) \varphi^*, \varphi^* \right\rangle  + \epsilon q \left\langle V^\square \varphi^*, \varphi^* \right\rangle  \\
& \geq g \|\varphi^* \|^2 - O(\epsilon)\|\varphi^*\|^2 \geq 0,
\end{align}
where the first inequality relates to the spectral gap $g$ of the kinetic energy 
and the norm $\|V\|$ of the single site perturbation, 
 and the last inequality holds for $\epsilon$ small enough. This finishes the proof.
\findem
\subsection{Positive random variables}
We study in this subsection the case involving positive random variables. 
We remind the reader of the 
the definition of the constants involved, for which we use the functions $\psi_i$, the matrix $A$, 
its eigenvalues $P_i$ and the linear space $\mathcal V_0$, which can be found at the beginning of this section. 
We define the subspace $\mathcal V_{01} \subset \mathcal V_{0}$ as
\begin{equation*} 
  \mathcal V_{01} := \spanop \limits_{\{i:P_i = P_1\}} \langle \psi_i \rangle,
\end{equation*}
i.e. the eigenspace of $A$ associated to its minimal eigenvalue $P_1$.

We recall the following quantities :
\begin{align*} 
  A'_1:=& \inf_{q \in S_\mu} \inf_{\substack{\psi \in \mathcal V_{0} \\ \|\psi\|_{\ell^2(\square)} = 1}} 
   q \left\langle \psi, V^\square \psi \right\rangle 
\\
=& \inf_{q \in S_\mu} \inf_{\substack{\psi \in \mathcal V_{01} \\ \|\psi\|_{\ell^2(\square)} = 1}} 
   q \left\langle \psi, V^\square \psi \right\rangle
= \min(s_+ P_1, s_- P_1) \in \R,
\end{align*}
and
\begin{equation*} 
A'_2:= - s_+^2 \sup_{\substack{\psi \in \mathcal V_{01}
    \\ \|\psi\|_{\ell^2(\square)} = 1}} \sup_{\substack{\varphi \in \mathcal V_0^\perp
    \\ \|\varphi\|_{\ell^2(\square)} = 1}} \frac{\left| \left\langle \psi,
  V^\square \varphi \right\rangle \right|^2}{\left\langle H^\square_0(\theta)
  \varphi, \varphi \right\rangle} \leq 0.
\end{equation*}
Note that, unlike the coefficient $A_1$ in the case of sign-changing random variables, in this case $A'_1$ may take on both signs. 
We also restate theorem \ref{thm:mainaltversionintro} for the reader's convenience.

\begin{thm}
\label{thm:mainaltversion}
Assume \eqref{hypa}, \eqref{hypb} and \eqref{hypcprime}. Fix $\theta \in \Theta$. 
Then, for $\epsilon>0$ small enough, if $P_1 \neq 0$,
\begin{equation}
  \bottom \leq \epsilon A'_1, \nonumber
\end{equation}
whereas if $P_1=A'_1 = 0$, but $A'_2\neq 0$, then
\begin{equation}
  \bottom \leq \epsilon^2 A'_2 + O(\epsilon^3). \nonumber
\end{equation}
Finally, if $P_1=A'_1 = A'_2 = 0$, then
$$\bottom \leq 0.$$
\end{thm}
The proof of this theorem is very similar to the proof of theorem \ref{thm:main}. 
Indeed, lemma \ref{lem:upperboundperiodic} is also valid in this setting. The theorem is then a consequence of the following lemma.
\begin{lem} \label{lem:infspecalt}
  Let $A'_1$ and $A'_2$ as in \eqref{def:Aprime1}, \eqref{def:Aprime2}, assume
  \eqref{hypa}, \eqref{hypb}, and \eqref{hypcprime} and fix $\theta \in \Theta$.  Then, for $\epsilon>0$ small enough, if $P_1 \neq 0$,
\begin{equation*}
  \inf_{q \in S_\mu} \inf_{\|u\|_{\ell^2(\square)} = 1} \left \langle H_{q, \epsilon}^\square(\theta) u,u \right \rangle \leq \epsilon A'_1,
\end{equation*}
whereas if $P_1 = A'_1 = 0$, but $A'_2\neq 0$, then
\begin{equation*}
  \inf_{q \in S_\mu} \inf_{\|u\|_{\ell^2(\square)} = 1} \left \langle H_{q, \epsilon}^\square(\theta) u,u \right \rangle \leq \epsilon^2 A'_2 + O(\epsilon^3)
\end{equation*}
Finally, if $P_1=A'_1 = A'_2 = 0$, then
\begin{equation*}
  \inf_{q \in S_\mu} \inf_{\|u\|_{\ell^2(\square)} = 1} \left \langle H_{q, \epsilon}^\square(\theta) u,u \right \rangle \leq 0.
\end{equation*}
\end{lem}

\subsubsection*{Sketch of proof}
We proceed in the argument as in lemma \ref{lem:infspec} up to equation \eqref{eq:expansion}. 
If $P_1 \neq 0$ we let 
\
\begin{equation*}
  u= \psi_1, \quad \varphi =0, \quad 
\end{equation*}
in  \eqref{eq:expansion} and thus  
\begin{equation*} 
  \inf_{q \in S_\mu} \inf_{\|u\|_{\ell^2(\square)} = 1} \left \langle H_{q, \epsilon}^\square(\theta) u,u \right \rangle 
\leq \inf_{q \in S_\mu} \left \langle H_{q, \epsilon}^\square(\theta) \psi_1,\psi_1 \right \rangle 
\end{equation*}
Choosing
\begin{equation*}
q = 
  \left\{
    \begin{matrix}
      s_+ & \textrm{ if } & P_1 < 0  \\
      s_- & \textrm{ if } & P_1 \geq 0 
    \end{matrix}
  \right.
\end{equation*}
we obtain \begin{equation} \label{eq:upperbound-fixedsign}
\inf_{q \in S_\mu} \left \langle H_{q, \epsilon}^\square(\theta) \psi_1,\psi_1 \right \rangle 
\leq \epsilon A'_1,
\end{equation}
If $P_1=0$ (and thus $A'_1=0$) but $A'_2 \neq 0$, 
then we find $\psi^* \in \mathcal V_{01}$ and $\varphi \in \mathcal V_0^\perp$ realizing the supremum in the definition of $A'_2$ and then we proceed as in lemma \ref{lem:infspec}.
In particular, we know 
\begin{equation} \label{eq:upperbound-fixedsign2}
  \inf_{q \in S_\mu} \inf_{\|u\|_{\ell^2(\square)} = 1} \left \langle H_{q, \epsilon}^\square(\theta) u,u \right \rangle 
\leq \epsilon^2 A'_2 + O(\epsilon^3)
\end{equation}
Finally, if $P_1=A'_1 = A'_2 = 0$ we take $u=\psi$ in \eqref{eq:expansion}
and conclude
\begin{equation} \label{eq:upperbound-fixedsign3}
  \inf_{q \in S_\mu} \inf_{\|u\|_{\ell^2(\square)} = 1} \left \langle H_{q, \epsilon}^\square(\theta) u,u \right \rangle 
\leq \epsilon A'_1= 0.
\end{equation}
\findem

We prove the following converse lemma.
\begin{lem} \label{lem:pos-convinfspec}
  Let $A'_1$ and $A'_2$ as in \eqref{def:Aprime1}, \eqref{def:Aprime2}, assume
  \eqref{hypa}, \eqref{hypb} and \eqref{hypcprime}, and fix $\theta \in \Theta$.  Then, for $\epsilon>0$ small enough, if $P_1 \neq 0$,
\begin{equation*}
  \inf_{q \in S_\mu} \inf_{\|u\|_{\ell^2(\square)} = 1} \left \langle H_{q, \epsilon}^\square(\theta) u,u \right \rangle \geq \epsilon A'_1 + O(\epsilon^{3/2}),
\end{equation*}
whereas if $P_1=A'_1 = 0$, but $A'_2 \neq 0$, then
\begin{equation*}
  \inf_{q \in S_\mu} \inf_{\|u\|_{\ell^2(\square)} = 1} \left \langle H_{q, \epsilon}^\square(\theta) u,u \right \rangle \geq \epsilon^2 A'_2 + O(\epsilon^3)
\end{equation*}
Finally, if $P_1=A'_1 = A'_2 = 0$, then
\begin{equation*}
  \inf_{q \in S_\mu} \inf_{\|u\|_{\ell^2(\square)} = 1} \left \langle H_{q, \epsilon}^\square(\theta) u,u \right \rangle \geq 0.
\end{equation*}
\end{lem}
\debdem
We adapt here the proof of lemma \ref{lem:convinfspec}.
Fix $\epsilon >0$ and let $q_\epsilon \in S_\mu$ be a value which minimizes the map
\begin{equation}
\label{eq:infimum2}
  q \mapsto \inf_{\|u\|_{\ell^2(\square)} = 1} \left \langle H_{\epsilon,
  q}^\square(\theta) u,u \right \rangle.
\end{equation}
We know a-priori $q_\epsilon \in [s_-,s_+]$. 
For simplicity, we write in the sequel simply $q$ for $q_\epsilon$.
We lower bound the quadratic form
by minimizing over a larger set by writing
\
\begin{multline*}
\inf_{\|u\|_{\ell^2(\square)} = 1} \left \langle ( H_{\epsilon,
  q}^\square(\theta) - \epsilon A'_1) u,u \right \rangle 
\\ \geq \inf_{\substack{\psi \in
    \mathcal V_{0} \\ \|\psi\|_{\ell^2(\square)} \leq 1}}
\inf_{\substack{\varphi \in \mathcal V_{0}^\perp
    \\ \|\varphi\|_{\ell^2(\square)} \leq 1}} \left \langle ( H_{\epsilon,
  q}^\square(\theta)  - \epsilon A'_1)(\psi+\varphi),(\psi+\varphi) \right \rangle. 
\end{multline*}
By continuity and compactness, there exists some pair 
$(\psi^*,\varphi^*):=(\psi^*_\epsilon,\varphi^*_\epsilon)$ in $\mathcal V_{0} \times \mathcal V_{0}^\perp$ realizing the infimum on the right hand
side. 
We see that
\begin{equation} \label{eq:pos-normvarphi}
  \left\langle H^\square_0(\theta)
  (\psi^* + \varphi^*), (\psi^* + \varphi^*) \right\rangle = \left\langle H^\square_0(\theta)
  \varphi^*, \varphi^* \right\rangle \geq g \|\varphi^*\|_{\ell^2(\square)}^2,
\end{equation}
where the constant $g$ is the spectral gap of $H^\square_0$, which is also its (positive) second eigenvalue. 

We study the different cases.

\subsubsection*{Case $P_1 \neq 0$.}
By (\ref{eq:upperbound-fixedsign}) in the previous lemma, we know
\begin{equation} \label{eq:pos-negspec}
  \left \langle (H_{q, \epsilon}^\square(\theta) - \epsilon A'_1) (\psi^*+\varphi^*),(\psi^*+\varphi^*) \right \rangle \leq 0.
\end{equation}
Using \eqref{eq:pos-normvarphi} and \eqref{eq:pos-negspec}  we get
\begin{equation*}
   \|\varphi^*\|^2_{\ell^2(\square)} 
\leq - \frac{\epsilon }{g} \left\langle (qV^\square -  A'_1) (\psi^* + \varphi^*), (\psi^* + \varphi^*) \right\rangle 
\leq 2\epsilon g^{-1}  \|qV^\square-A'_1\| .
\end{equation*}
Note that 
\
\begin{align*}
   q \left\langle V^\square \psi^*, \psi^* \right\rangle 
&\geq 
\|\psi^*\|^2 q \inf_{\substack{\psi \in \mathcal V_0 \\ \|\psi\|_{\ell^2(\square)} \leq 1}} \left\langle V^\square \psi, \psi \right\rangle  
\\
&=  \|\psi^*\|^2 q \inf_{\substack{\psi \in \mathcal V_{01} \\ \|\psi\|_{\ell^2(\square)} \leq 1}} \left\langle V^\square \psi, \psi \right\rangle 
\geq \|\psi^*\|^2 A'_1.
\end{align*}
and thus
\begin{equation*}
    \left\langle  (q V^\square - A'_1) \psi^*, \psi^* \right\rangle \geq 0.
\end{equation*}
On the other hand, we have $ \|qV^\square-A'_1\|\leq  2s_+\|V^\square\| $. This implies
\begin{multline*}
  \inf_{\|u\|_{\ell^2(\square)} = 1} \left \langle (H_{q, \epsilon}^\square(\theta) -\epsilon A'_1) u,u \right \rangle \geq  
  \epsilon \left\langle (q V^\square - A'_1) \psi^*, \psi^* \right\rangle \\
  \phantom{\widehat MMMMMMMM} + 2 \epsilon \Re \left\langle (q V^\square - A'_1)  \varphi^*, \psi^* \right\rangle  + \epsilon \left\langle (q V^\square - A'_1)  \varphi^*, \varphi^* \right\rangle 
\\
 \phantom{\widehat M} \geq  - 2^{3/2}  g^{-1/2}\epsilon^{3/2}  \|qV^\square-A'_1\|^{3/2}- 2 g^{-1}\epsilon^2  \|qV^\square-A'_1\|^2 
\\
 \geq  - 8  g^{-1/2}\epsilon^{3/2}  s_+^{3/2}\|V^\square\|^{3/2}- 8 g^{-1}\epsilon^2 s_+^2  \|V^\square\|^2 .
\end{multline*}
\subsubsection*{Case $P_1 = A'_1 = 0$ and $A'_2 \neq 0$.}
We know from 
\eqref{eq:upperbound-fixedsign2}
that
\begin{equation} \label{eq:pos-negspec2}
0>   \left \langle H_{q, \epsilon}^\square(\theta) (\psi^*+\varphi^*),(\psi^*+\varphi^*) \right \rangle.
\end{equation}
 We will decompose further $\psi^* = \psi^*_{01} + \psi^*_{0\perp} \in
 \mathcal V_0$, with $\psi^*_{01} \in \mathcal V_{01}$ and
 $\psi^*_{0\perp} \in \mathcal V_{01}^\perp$. 
Using \eqref{eq:orthogonalityrelations} and $P_1=0$, we conclude
\
\begin{gather*}
\left\langle  V^\square \psi^*_{01}, \psi^*_{01} \right\rangle  
= \left\langle  V^\square \psi^*_{01}, \psi^*_{0\perp} \right\rangle  = 0
\\
\text{and }\left\langle  V^\square \psi^*_{0\perp}, \psi^*_{0\perp} \right\rangle  
\geq 0
\end{gather*}
Hence
\begin{align}
\bigl\langle  V^\square (\psi^*+\varphi^*), (&\psi^*+\varphi^*) \bigr\rangle\label{eq:quadratic_form} \\
= 
& \left\langle  V^\square \psi^*_{01}, \psi^*_{01} \right\rangle  + 2 \Re \left\langle  V^\square \psi^*_{01}, \psi^*_{0\perp} \right\rangle   + \left\langle  V^\square \psi^*_{0\perp}, \psi^*_{0\perp} \right\rangle  \nonumber 
\\
& + 2 \Re \left\langle  V^\square \psi^*_{01}, \varphi^* \right\rangle 
+ 2 \Re \left\langle  V^\square \psi^*_{0\perp}, \varphi^* \right\rangle 
+ \left\langle  V^\square \varphi^*, \varphi^* \right\rangle \nonumber \\
\geq
& 2 \Re \left\langle  V^\square \psi^*_{01}, \varphi^* \right\rangle 
+  2 \Re \left\langle  V^\square \psi^*_{0\perp}, \varphi^* \right\rangle   
+ \left\langle  V^\square \varphi^*, \varphi^* \right\rangle  \nonumber
\end{align}
In the specific case $\Re \left\langle  V^\square \psi^*_{0\perp}, \varphi^* \right\rangle =0$ 
we obtain:
\begin{align} \nonumber
\left\langle  V^\square (\psi^*+\varphi^*), (\psi^*+\varphi^*) \right\rangle 
&\geq
  2 \Re \left\langle  V^\square \psi^*_{01}, \varphi^* \right\rangle   
+ \left\langle  V^\square \varphi^*, \varphi^* \right\rangle  \label{ineq:magic}.
\\&
  \geq
  2 \Re \left\langle  V^\square \psi^*_{01}, \varphi^* \right\rangle   
 - \| V^\square\|  \| \varphi^*\|^2.
\end{align}
Using this bound, \eqref{eq:pos-normvarphi}  and   (\ref{eq:pos-negspec2}), we obtain
\begin{align*} 
0&>  
\left \langle H_{q, \epsilon}^\square(\theta) (\psi^*+\varphi^*),(\psi^*+\varphi^*) \right\rangle \\
  &\geq  \left\langle H^\square_0(\theta)  \varphi^*, \varphi^* \right\rangle + 2 \epsilon q  \Re \left\langle  V^\square \psi^*_{01}, \varphi^* \right\rangle 
+ \epsilon q \left\langle  V^\square  \varphi^*, \varphi^* \right\rangle 
\\
  &\geq  g \|\varphi^*\|^2 + 2 \epsilon q  \Re \left\langle  V^\square \psi^*_{01}, \varphi^* \right\rangle 
- \epsilon q \| V^\square\|  \| \varphi^*\|^2
\end{align*}
Therefore 
\begin{equation*}
   \|\varphi^*\|_{\ell^2(\square)}^2 
\leq  \epsilon q g^{-1} \|V^\square\|  ( 2  \|\varphi^*\|_{\ell^2(\square)} +  \|\varphi^*\|_{\ell^2(\square)}^2) 
\leq  3 \epsilon q g^{-1} \|V^\square\|  \|\varphi^*\|_{\ell^2(\square)}.
\end{equation*}
which simplifies to  
\begin{equation} \label{ineq:pos-smallnorm}
   \|\varphi^*\|_{\ell^2(\square)} \leq 3 \epsilon q g^{-1} \|V^\square\| .
\end{equation}
This inequality implies 
\begin{align} 
\nonumber
0&> \left \langle H_{q, \epsilon}^\square(\theta) (\psi^*+\phi^*),(\psi^*+\phi^*) \right \rangle 
\\ 
\label{eq:beforelambda}
&\geq 
  \left\langle H^\square_0(\theta)  \varphi^*, \varphi^* \right\rangle 
+ 2 \epsilon q  \Re \left\langle  V^\square \psi^*_{01}, \varphi^* \right\rangle 
- \frac{9}{g^2} \epsilon^3 q^3  \|V^\square\|^3
\\
\nonumber
&\geq 
g \|\varphi^*\|_{\ell^2(\square)}^2
+ 2 \epsilon q  \Re \left\langle  V^\square \psi^*_{01}, \varphi^* \right\rangle 
- \frac{9}{g^2} \epsilon^3 q^3  \|V^\square\|^3
\end{align}
and we see that $\left\langle V^\square \psi^*_{01}, \varphi^* \right\rangle \neq 0$ 
for small $\epsilon$. Thus the choice 
$$ 
\lambda = - \frac{\left\langle H_0^\square(\theta) \varphi^*, \varphi^* \right\rangle}{\left\langle V^\square \psi^*_{01}, \varphi^* \right\rangle}.
$$
is well defined for small $\epsilon$. 
We multiply (\ref{eq:beforelambda}) by $1=|\lambda|^2/|\lambda|^2=\overline \lambda/\overline \lambda$
and obtain
\begin{align} \label{eq:afterlambda}
 0&> \left \langle H_{q, \epsilon}^\square(\theta) (\psi^*+\phi^*),(\psi^*+\phi^*) \right \rangle 
\\ \nonumber
 &\geq   |\lambda|^2 \frac{\left\langle H_0^\square(\theta) \varphi^*, \varphi^* \right\rangle}{|\lambda|^2}  +  2 \Re  \overline \lambda  \epsilon q   \frac{\left\langle V^\square \psi^*_{01}, \varphi^* \right\rangle}{\overline \lambda} 
- \frac{9}{g^2} \epsilon^3 q^3  \|V^\square\|^3
\\ \nonumber
&\geq (|\lambda|^2  -  2 \Re \overline \lambda \epsilon q)  \frac{\left| \left\langle V^\square \psi^*_{01}, \varphi^* \right\rangle \right|^2}{\left\langle H_0^\square(\theta) \varphi^*, \varphi^* \right\rangle} +   O(\epsilon^3) 
\\ \nonumber
&\geq
-\epsilon^2 q^2 \frac{\left| \left\langle V^\square \psi^*, \varphi^* \right\rangle \right|^2}{\left\langle H_0^\square(\theta) \varphi^*, \varphi^* \right\rangle} +   O(\epsilon^3) 
\end{align}
where in the last line we used the trivial bound
$|\lambda|^2 - 2 \Re \overline \lambda \epsilon q 
\geq |\lambda|^2 - 2 |\lambda| \epsilon q \geq -\epsilon^2 q^2$. 
Since $-q^2 \geq -s_+^2$, this implies by the very definition of $A'_2$
\begin{align*} 
 0 &> \left \langle H_{q, \epsilon}^\square(\theta) (\psi^*+\phi^*),(\psi^*+\phi^*) \right \rangle 
\geq  \epsilon^2  A'_2 +  O(\epsilon^3).
\end{align*}
In the alternative case that 
$\Re \left\langle V^\square \psi^*_{0\perp}, \varphi^* \right\rangle \neq 0$ the coefficient
\begin{equation*}
  \mu 
:= - \frac{\left\langle V^\square \psi_{0\perp}^*, \psi_{0\perp}^* \right\rangle}{\left\langle V^\square \psi^*_{0\perp}, \varphi^* \right\rangle}
\in \C
\end{equation*} 
is well defined. Moreover, it implies $\psi^*_{0\perp} \neq 0$. Consequently, 
$\mathcal V_0 \neq \mathcal V_{01}$ and there exists an eigenvalue of the matrix $A$ strictly
larger than $P_1$. Thus $g_A:= \min\{P_i \mid i =1, \ldots, p, P_i >P_1 \}$ is finite and strictly 
positive and 
$\left\langle V^\square \psi_{0\perp}^*, \psi_{0\perp}^* \right\rangle 
\geq g_A \|\psi_{0\perp}^* \|^2$. 
We can now argue 
\begin{align*}
&\left\langle  V^\square (\psi^*+\varphi^*), (\psi^*+\varphi^*) \right\rangle 
\\
\geq  
&\left\langle  V^\square \psi^*_{0\perp}, \psi^*_{0\perp} \right\rangle
+ 2 \Re \left\langle  V^\square \psi^*_{01}, \varphi^* \right\rangle 
+ 2 \Re \left\langle  V^\square \psi^*_{0\perp}, \varphi^* \right\rangle   
+ \left\langle  V^\square \varphi^*, \varphi^* \right\rangle  
\\  \nonumber 
=&   |\mu|^2 \frac{ \left\langle  V^\square \psi^*_{0\perp}, \psi^*_{0\perp} \right\rangle} {|\mu|^2} +  2 \Re \overline \mu \frac{  \left\langle  V^\square \psi^*_{0\perp}, \varphi^* \right\rangle} {\overline \mu} 
\\
 &\phantom{ |\mu|^2 \frac{ \left\langle  V^\square \psi^*_{0\perp}, \psi^*_{0\perp} \right\rangle} {|\mu|^2}} +2 \Re \left\langle  V^\square \psi^*_{01}, \varphi^* \right\rangle 
+ \left\langle  V^\square \varphi^*, \varphi^* \right\rangle \nonumber 
\\
=&  (|\mu|^2  -  2 \Re \overline \mu)  \frac{\left| \left\langle V^\square \psi^*_{0\perp}, \varphi^* \right\rangle \right|^2}{\left\langle V^\square \psi^*_{0\perp}, \psi^*_{0\perp} \right\rangle}
 +  2 \Re \left\langle  V^\square \psi^*_{01}, \varphi^* \right\rangle 
 + \left\langle  V^\square \varphi^*, \varphi^* \right\rangle   \nonumber \\
\geq & -\frac{\left| \left\langle V^\square \psi^*_{0\perp}, \varphi^* \right\rangle \right|^2}{\left\langle V^\square \psi^*_{0\perp}, \psi^*_{0\perp} \right\rangle}
 + 2 \Re \left\langle  V^\square \psi^*_{01}, \varphi^* \right\rangle 
+ \left\langle  V^\square \varphi^*, \varphi^* \right\rangle   \nonumber  \\
\end{align*}
Now we estimate
\begin{align*}
\frac{\left| \left\langle V^\square \psi^*_{0\perp}, \varphi^* \right\rangle \right|^2}{\left\langle V^\square \psi^*_{0\perp}, \psi^*_{0\perp} \right\rangle}
\leq
\frac{\left\|V^\square\right\|^2 \, \left\|\psi^*_{0\perp}\right\|^2 \, \left\|\varphi^* \right\|^2}
{g_A \left\|\psi^*_{0\perp}\right\|^2}
=
\frac{\left\|V^\square\right\|^2 \, \left\|\varphi^* \right\|^2}{g_A}
\end{align*}
and conclude, analogously to \eqref{ineq:magic},
\begin{align*}
\bigl\langle  V^\square (\psi^*+\varphi^*), (&\psi^*+\varphi^*) \bigr\rangle \\
\geq  
 &  2 \Re \left\langle  V^\square \psi^*_{01}, \varphi^* \right\rangle 
+ \left\langle  V^\square \varphi^*, \varphi^* \right\rangle   \nonumber  
- \frac{\left\|V^\square\right\|^2 \, \left\|\varphi^* \right\|^2}{g_A}
\\
\geq  
 &  2 \Re \left\langle  V^\square \psi^*_{01}, \varphi^* \right\rangle 
\nonumber  
-\left(\left\|V^\square\right\|   + \frac{\left\|V^\square\right\|^2 }{g_A} \right ) \left\|\varphi^* \right\|^2.
\end{align*}
This implies two bounds similar to \eqref{ineq:pos-smallnorm}
and \eqref{eq:beforelambda}, respectively
\[
\left\|\varphi^* \right\| \leq 
\frac{\epsilon q}{g} \left[3 \|V^\square\| + \frac{\|V^\square\|^2}{g_A} \right] 
\]
and
\[ 
 0>  g \|\varphi^*\|^2 
+ 2 \epsilon q  \Re \left\langle  V^\square \psi^*_{01}, \varphi^* \right\rangle 
- O(\epsilon^3 q^3  ).
\]
The proof is now concluded as in \eqref{eq:afterlambda}.

\subsubsection*{Case $P_1=A'_1 = 0$ and $A'_2 = 0$.} 
The last condition implies
\begin{equation} 
\sup_{\substack{\psi \in \mathcal V_{01}
    \\ \|\psi\|_{\ell^2(\square)} = 1}} \sup_{\substack{\varphi \in \mathcal V_0^\perp
    \\ \|\varphi\|_{\ell^2(\square)} = 1}} 
\left| \left\langle \psi,
  V^\square \varphi \right\rangle \right|^2= 0.
\end{equation}
and in particular 
$\left\langle  V^\square \psi^*_{01}, \varphi^* \right\rangle =0$.
Similarly as in the previous case, a calculation completing the square will be helpful.
For small $\epsilon$, we have
\begin{align*}
&(g -\epsilon q \|V^\square\|)\left \| \varphi^* \right\|^2
+  2 \epsilon q \Re \left \langle  V^\square \psi^*_{0\perp} , \varphi^* \right\rangle 
\\ 
&\geq \frac{1}{2} g \left \| \varphi^* \right\|^2
-  2 \epsilon q \left \|  V^\square \right\| \left \| \psi^*_{0\perp}  \right\| \left \| \varphi^* \right\|
\geq -2 \epsilon^2 q^2 g^{-1}  \|  V^\square \|^2     \|\varphi^* \|^2 
 \end{align*}
We now use (\ref{eq:quadratic_form}) and $\left\langle V^\square \psi_{0\perp}^*, \psi_{0\perp}^* \right\rangle 
\geq g_A \|\psi_{0\perp}^* \|^2$ again, assuming for the moment that $ \psi^*_{0\perp}\neq 0$ and thus 
$g_A>0$.
Then
\begin{align*} 
&\inf_{\|u\|_{\ell^2(\square)}=1} \left \langle H_{q, \epsilon}^\square(\theta) u,u \right \rangle 
\\ 
&\geq (g -\epsilon q \|V^\square\|)\left\| \varphi^* \right\|^2
+  2 \epsilon q \Re \left\langle  V^\square  \psi^*_{0\perp}, \varphi^* \right\rangle + \epsilon q \langle V^\square  \psi^*_{0\perp} , \psi^*_{0\perp}\rangle 
\\
&\geq - 2 \epsilon^2 q^2 g^{-1}\left\|  V^\square \right\|^2 \left\|  \psi^*_{0\perp} \right\|^2 + \epsilon q g_A \| \psi^*_{0\perp} \|^2 \geq 0
\end{align*}
for $\epsilon>0$ small enough. If $ \psi^*_{0\perp}= 0$ then the eigenvalue gap $g_A$ 
does not exist, but then we have an even better lower bound
\begin{align*} 
\inf_{\|u\|_{\ell^2(\square)}=1} \left \langle H_{q, \epsilon}^\square(\theta) u,u \right \rangle 
\geq (g -\epsilon q \|V^\square\|)\left\| \varphi^* \right\|^2
\end{align*}

\findem

\subsection{Application to the discrete alloy type model}
\subsubsection*{Proof of theorem \ref{thm:anderson}}
It is enough
to verify that the assumptions of theorem \ref{thm:main} are
satisfied. Let $C_W:= \inf \sigma(-\Delta_{\Z^d} + W)$. It is clear
that the operator $H_0 := -\Delta_{\Z^d} + W - C_W $ satisfies hypothesis
\eqref{hypa}. In this case the set $\Theta$ consists of the single point $\theta=0$ (see theorem \ref{thm:kirschsimon} below).

Let us check property \eqref{cond:nonzero} for the operator
\
\begin{equation*}
  H_0(0) : = -\Delta_{\square} + W_\square + C_W,
\end{equation*} 
where $\Delta_{\square}$ is the Laplacian on $\square$ with periodic boundary conditions and $W_\square$ is the restriction of $W$ to $\square$.
To check the property we
use Perron--Frobenius theorem \cite{meyer2000}. For $m > |\square|$,
we verify that
\begin{equation*}
  \left\langle \delta_x , \left(\Delta_\square - W_\square + \|W\|_{\infty}  + 2d + 1 \right)^m \delta_y \right\rangle 
\geq 1.
\end{equation*}
This implies that the
largest eigenvalue of the matrix $\Delta_\square - W_\square + \|W\|_{\infty} + 2d + 1$ is simple and its corresponding eigenfunction
$\psi_1$ is positive (i.e. $(\forall n\in \square) \, \psi_1(n)
>0$). Because of this strict positivity, condition
\eqref{cond:nonzero} is satisfied as soon as $V^\square \not \equiv
0$. The subspace $\mathcal V_0$ is thus one-dimensional and contains only $\psi_1$. The theorem is now proven, by simply stating the consequences of theorem \ref{thm:main}.  \findem

We know recall theorem 2.4 in \cite{kirschsimon1987}, with our
notations. It implies that $0$ is the unique $\theta \in \square^*$
realizing the minimum of the spectrum.
\begin{thm}
\label{thm:kirschsimon}
Let $H_0 = -\triangle_{\Z^d} + W$ with $W$ a periodic potential with respect to $\gamma = N\Z^d$, and $E_0(\theta)$ be the smallest eigenvalue of $H_0(\theta)$. Then
$$\left(a_-/a_+ \right)^2 \left( 2d - \sum_{i=1}^d \cos(\theta_i) \right) \leq E_0(\theta) - E_0(0) \leq \left( 2d - \sum_{i=1}^d \cos(\theta_i) \right).$$
Here, $a_{\pm} = \pm \max \pm \psi_1$ and $\psi_1$ is the positive ground state of $H_0(0)$.
\end{thm}

\section{Appendix}
\subsection{An interesting example: Proof of theorem \ref{thm:interesting}}
Let $H_0 := \Delta_{\Z}^2$ defined on $\ell^2(\Z)$. This operator has hopping range $N=3$ (see \eqref{hypa}) and thus $\square=\{-1, 0,1,\}$. We define $V^\square$ as the multiplication operator given by the following single-site potential:
\
\begin{align*}
  V^\square & : \ell^2(\square) \to \R \\
   V^\square(n) &:= -\frac{1}{2} \delta_{-1}(n) + \delta_0(n) - \frac{1}{2} \delta_1(n).
\end{align*}

With these definitions, we see that, for $\theta \in [-\pi/3, \pi/3)^d$,
\
\begin{equation*}
  H^\square_0(\theta) = 
  \left(
  \
  \begin{matrix}
    6 & -4 + e^{-3i\theta}& 1-4e^{-3i\theta}  \\
    -4 + e^{3i\theta}& 6 & -4 + e^{-3i\theta} \\
     1 - 4e^{-3i\theta}&  -4 + e^{3i\theta}& 6
  \end{matrix}
  \right),
\end{equation*}
after \eqref{def:floquetmatrix}. This matrix has a simple ground state 
\begin{equation*}
  \psi_0(\theta):=(e^{-i\theta},1,e^{i\theta})/\sqrt 3
\end{equation*}
with eigenvalues $E_0(\theta)=(2-2\cos(\theta))^2$. Let now $\tilde f_n(\theta) = \chi_n \tilde \psi_0(\theta) \in \ell^2(\Z^d)$ where $ \tilde \psi_0(\theta)$ is the $\theta$-quasi-$\gamma$-periodic extension of $\psi_0(\theta)$. Finally, for $\xi > 1/4$, let
\
\begin{equation*}
  u_n := f_n(0) + \epsilon^\xi f_n(\epsilon^\xi).
\end{equation*}
Let us calculate the kinetic energy. We see that
\
\begin{multline}\label{eq:exdev}
  \langle H_0 u_n, u_n \rangle 
\\ = \langle H_0 f_n(0), f_n(0) \rangle + 2 \epsilon^\xi \Re \langle H_0 f_n(0) , f_n(\epsilon^\xi) \rangle + \epsilon^{2\xi} \langle H_0 f_n( \epsilon^\xi) , f_n(\epsilon^\xi) \rangle.
\end{multline}
Let $\delta >0$ and pick $n$ so large so that
\
\begin{equation*}
 \left|  \frac{\left\langle f_n(0) , H_0 f_n(0) \right \rangle_{\ell^2(D)}}{\left\| f_n(0) \right\|_{\ell^2(D)}^2} - \frac{\left\langle \psi_0(0), H_0^\square(\theta) \psi_0(0) \right\rangle_{\ell^2(\square)}}{\left\|  \psi_0(0) \right\|_{\ell^2(\square)}^2} \right| \leq \delta,
\end{equation*}

\begin{equation*}
 \left|  \frac{\left\langle f_n(\epsilon^\xi) , H_0 f_n(\epsilon^\xi) \right \rangle_{\ell^2(D)}}{\left\| f_n(\epsilon^\xi) \right\|_{\ell^2(D)}^2} - \frac{\left\langle \psi_0(\epsilon^\xi ), H_0^\square(\theta) \psi_0(\epsilon^\xi ) \right\rangle_{\ell^2(\square)}}{\left\|  \psi_0(\epsilon^\xi ) \right\|_{\ell^2(\square)}^2} \right| \leq \delta,
\end{equation*}
and
\begin{equation*}
 \left| \frac{\left\| H_0 f_n(0) \right\|^2_{\ell^2(D)}} {\left\| f_n(0) \right\|_{\ell^2(D)}^2}  \right|
= \left|  \frac{\left\langle f_n(0) , H_0^2 f_n(0) \right \rangle_{\ell^2(D)}}{\left\| f_n \right\|_{\ell^2(D)}^2}  \right| \leq \delta.
\end{equation*}
Then, from \eqref{eq:exdev} we see that
\
\begin{multline*}
    \langle H_0 u_n, u_n \rangle 
\\ \leq \delta\| f_n(0)\|^2 + 2 \epsilon^\xi \delta \| f_n(\epsilon^\xi) \| + \epsilon^{2\xi} \left\langle \psi_0(\epsilon^\xi ), H_0^\square(\theta) \psi_0(\epsilon^\xi ) \right\rangle_{\ell^2(\square)} \| f_n(\theta) \|^2 
\\  \leq 3 \delta + \epsilon^{2\xi} E_0(\epsilon^\xi) \| f_n(\theta) \|^2.
\end{multline*}
Letting $n\to \infty$ and $\delta \to 0$ we see that
\begin{equation*}
    \langle H_0 u_n, u_n \rangle \leq  \epsilon^{2\xi} E_0(\epsilon^\xi) \leq C \epsilon^{6 \xi} .
\end{equation*}
Now let us calculate the potential energy.
\begin{multline*}\label{eq:exdev}
  \epsilon \langle V_q u_n, u_n \rangle \\ 
= \epsilon \langle V_q f_n(0), f_n(0) \rangle + 2  \epsilon^{1 + \xi} \Re \langle V_q f_n(0) , f_n(\epsilon^\xi) \rangle + \epsilon^{1+2\xi} \langle V_q f_n( \epsilon^\xi) , f_n(\epsilon^\xi) \rangle \\
  =   2  \epsilon^{1 + \xi} \Re \langle V_q f_n(0) , f_n(\epsilon^\xi) \rangle .
\end{multline*}
 Now we can calculate explicitly
\
\begin{equation*}
   \langle V_q f_n( \epsilon^\xi) , f_n(\epsilon^\xi) \rangle = \frac{1}{6} \bigl( -e^{-i\epsilon^\xi} + 2 -e^{i\epsilon^\xi} \bigr) = - \frac{1}{3} \epsilon^{\xi} + O(\epsilon^{2\xi}).
\end{equation*}
This shows that, for small $\epsilon$,
\
\begin{equation*}
  \langle H_{\epsilon,q} u_n, u_n \rangle \leq C \xi^{6\xi}- \frac{1}{3} \epsilon^{1+2\xi} + O(\epsilon^{1+ 2\xi}) \leq -\frac{1}{6} \epsilon^{1+2\xi},
\end{equation*}
where we have used that $6 \xi > 1 + 2\xi$.
\label{sec:appendix}
\subsection{Proof of lemma \ref{lem:limitquasiperiodic}}
As the $V_q$ is block-diagonal, it is enough to do the
calculation for the free operator $H_0$. Let us first calculate some
norms. Because of the quasi-periodicity, we easily see that
\begin{align}
  \label{eq:normun}
  \left\|  u_n \right\|_{\ell^2(D)}^2 = (2n+1)^d \left\|  u_0 \right\|_{\ell^2(\square)}^2.
\end{align}
and
$$\|u_n - u_{n-1}\|_{\ell^2(D)}^2 \leq C n^{d-1} \left\|  u_0 \right\|_{\ell^2(\square)}^2 .$$
So we have that
\begin{align}
  \label{eq:diffun}
  \left\langle H_0 u_n, u_{n} \right\rangle = \left\langle H_0 u_{n},
u_{n-1} \right\rangle + O(n^{d-1})\|u_0\|_{\ell^2(\square)}^2.
\end{align}
 For any $k \in \square_{n-1}$ and $k' \in \Z^d \diagdown
 \square_{n}$, we have that $|k-k'| \geq N$ and thus, because of the
 finite hopping range (assumption \eqref{hypb}),
$$\left\langle H_0  u_{n}, u_{n-1} \right\rangle = \left\langle H_0  u, u_{n-1} \right\rangle.$$
Now, we develop
\begin{align}
\label{eq:kineticenergyquasiperiodic}
\bigl\langle & H_0  u, u_{n-1} \bigr\rangle   \nonumber
\\ & = \sum_{k \in \Z^d} \sum_{k' \in \Z^d}  H_0(k,k') u(k') \overline{ u_{n-1}(k)}
\\ & = \sum_{\substack{ \vphantom{m'} m \in \gamma \\  \vphantom{m'}  |m|\leq (n-1)N}} \sum_{\substack{ \vphantom{m'} m' \in \gamma }} \sum_{k \in \square + m} \sum_{k'  \in \square + m'}  H_0(k,k')u(k') \overline{ u_{n-1}(k)}
\nonumber
\\ & = \sum_{\substack{ \vphantom{m'} m \in \gamma \\  \vphantom{m'}  |m|\leq (n-1)N}} \sum_{\substack{ \vphantom{m'} m' \in \gamma }} \sum_{k \in \square } \sum_{k'  \in \square}  H_0(k+m,k'+m')u(k'+m') \overline{ u_{n-1}(k+m)}
\nonumber
\end{align}
Using the translation invariance (assumption \eqref{hypb}), the last quantity
is equal to
\begin{align}
\\ & \phantom { = } \sum_{\substack{ \vphantom{m'} m \in \gamma \\  \vphantom{m'}  |m|\leq (n-1)N}} \sum_{\substack{ \vphantom{m'} m' \in \gamma }} \sum_{k \in \square } \sum_{k'  \in \square}  H_0(k,k'+m'-m)u(k'+m'-m) \overline{ u_{n-1}(k)}
\nonumber
\\ & = \sum_{\substack{ \vphantom{m'} m \in \gamma \\  \vphantom{m'}  |m|\leq (n-1)N}} \sum_{\substack{ \vphantom{m'} m' \in \gamma }} \sum_{k \in \square } \sum_{k'  \in \square} e^{i \theta \cdot(m-m')}  H_0(k,k'-m'+m)u_0(k') \overline{ u_{0}(k)}
\nonumber
\\ & = \sum_{\substack{ \vphantom{m'} m \in \gamma \\  \vphantom{m'}  |m|\leq (n-1)N}} \sum_{\substack{ \vphantom{m'} m'' \in \gamma }} \sum_{k \in \square } \sum_{k'  \in \square} e^{i \theta \cdot m''}  H_0(k,k'-m'')u_0(k') \overline{ u_{0}(k)}
\nonumber
\\ & = (2n-3)^d \left\langle H_0(\theta)  u_0, u_0 \right\rangle \nonumber.
\end{align}
We see from this calculation and \eqref{eq:diffun} thus that
$$ \left| \left\langle H_0 u_n, u_n \right\rangle - (2n-3)^d \left\langle H_0(\theta)   u_0,  u_0 \right\rangle  \right| \leq  C n^{d-1} \|u_0\|_{\ell^2(\square)}^2.$$
As $(2n-3)/(2n-1) \to 1$, dividing by $\|u_n\|^2_{\ell^2(D)}$, using \eqref{eq:normun} and taking the limit proves the lemma.
\findem
\newpage

\thanks{
\section*{Acknowledgment}
The research of D.B. was supported by Russian Scientific Foundation (project no.  14-11-00078).
F.H-E and I.V.  were supported by DFG-Projects \emph{Zuf\"allige und periodische Quantengraphen} and 
\emph{Eindeutige-Fortsetzungsprinzipien und Gleichverteilungseigenschaften von Eigenfunktionen}.

We thank M. T\"aufer for reading a previous version of the manuscript and the referee for useful remarks.
}

\end{document}